\documentclass{lmcs}
\pdfoutput=1

\usepackage{lastpage}
\lmcsdoi{19}{2}{9}
\lmcsheading{}{\pageref{LastPage}}{}{}%
{Aug.~17,~2022}{May~08,~2023}{}

\usepackage{version}

\includeversion{rep}
\excludeversion{conf}

\newcommand\paper{article}
\newcommand\oursubsection[1]{\subsection{#1}}

\usepackage{cite}
\usepackage{etoolbox}
\usepackage{dcolumn}
\usepackage{algorithmicx}
\usepackage{regexpatch}
\usepackage{algorithm}
\usepackage[noend]{algpseudocode}
\usepackage{prftree}
\usepackage[utf8]{inputenc}
\usepackage{tikz}
\usepackage{pgfplots} 
\usepackage{amsmath}
\usepackage{amssymb}
\usepackage{booktabs}
\usepackage{stmaryrd}
\usepackage[nomessages]{fp}
\usepackage{xr}
\usepackage{stackengine}
\usepackage{calc}
\usepackage[mathscr]{eucal}
\usepackage[f]{esvect}
\xpatchcmd{\algorithmic}{\labelsep 0.5em}{\labelsep 3em}{\typeout{Success!}}{\typeout{Oh dear!}}

\title{SAT-Inspired Higher-Order Eliminations}

\author[J.~Blanchette]{Jasmin Blanchette\rsuper{a,b}\lmcsorcid{0000-0002-8367-0936}}	

\author[P.~Vukmirovi\'c]{Petar Vukmirovi\'c\rsuper{c}\lmcsorcid{0000-0001-7049-6847}}	

\address{\lsuper{a}Ludwig-Maximilians-Universität München, Institut für Informatik, Munich, Germany}
\email{jasmin.blanchette@ifi.lmu.de}
\address{\lsuper{b}Max-Planck-Institut f\"ur Informatik, Saarland Informatics Campus, Saarbr\"ucken, Germany}
\email{jasmin.blanchette@mpi-inf.mpg.de}
\address{\lsuper{c}Vrije Universiteit Amsterdam,
Department of Computer Science, Amsterdam,
The Netherlands}	
\email{petar.vukmirovic2@gmail.com}

\theoremstyle{definition}
\newtheorem{conventionx}[thm]{Convention}

\let\oldlabelitemi=\labelitemi
\let\labelitemi=\labelitemii
\let\labelitemii=\oldlabelitemi

\def\negvthinspace{\kern-0.083333em}
\def\vthinspace{\kern+0.083333em}
\def\vvthinspace{\kern+0.0416667em}
\def\negvvthinspace{\kern-0.0416667em}

\begin{document}

\newdimen\carpetH
\newdimen\carpetV
\def\carpet#1{\setbox0=\hbox{\ensuremath{#1}}
  \kern+2\carpetH
    \raise+2\carpetV\copy0\kern-\wd0
    \raise+\carpetV\copy0\kern-\wd0
    \copy0\kern-\wd0
    \raise-\carpetV\copy0\kern-\wd0
    \raise-2\carpetV\copy0\kern-\wd0
  \kern-\carpetH
    \raise+2\carpetV\copy0\kern-\wd0
    \raise+\carpetV\copy0\kern-\wd0
    \copy0\kern-\wd0
    \raise-\carpetV\copy0\kern-\wd0
    \raise-2\carpetV\copy0\kern-\wd0
  \kern-\carpetH
    \raise+2\carpetV\copy0\kern-\wd0
    \raise+\carpetV\copy0\kern-\wd0
    \copy0\kern-\wd0
    \raise-\carpetV\copy0\kern-\wd0
    \raise-2\carpetV\copy0\kern-\wd0
  \kern-\carpetH
    \raise+2\carpetV\copy0\kern-\wd0
    \raise+\carpetV\copy0\kern-\wd0
    \copy0\kern-\wd0
    \raise-\carpetV\copy0\kern-\wd0
    \raise-2\carpetV\copy0\kern-\wd0
  \kern-\carpetH
    \raise+2\carpetV\copy0\kern-\wd0
    \raise+\carpetV\copy0\kern-\wd0
    \copy0\kern-\wd0
    \raise-\carpetV\copy0\kern-\wd0
    \raise-2\carpetV\copy0
  \kern2\carpetH
}

\newcommand\heavy[1]{\carpetH=.02ex\carpetV=.02ex\carpet{#1}}
\newcommand\light[1]{\carpetH=.01ex\carpetV=.01ex\carpet{\scriptstyle#1}}

\newcommand\rotiota{\hbox{\rotatebox[origin=c]{180}{$\iota$}}}

\newcommand{\ifalse}{\heavy\bot}
\newcommand{\itrue}{\heavy\top}
\newcommand{\inot}{\heavy{\neg}}
\newcommand{\iand}{\mathbin{\heavy\land}}
\newcommand{\ior}{\mathbin{\heavy\lor}}
\newcommand{\iimplies}{\mathrel{\heavy{\rightarrow}}}
\newcommand{\iimplied}{\mathrel{\heavy{\leftarrow}}}
\newcommand{\iequiv}{\mathrel{\heavy\leftrightarrow}}
\newcommand{\iforall}{\heavy\forall}
\newcommand{\iexists}{\heavy\exists}
\newcommand{\iforallconst}{\heavy{\mathsf{\forall}}}
\newcommand{\iexistsconst}{\heavy{\mathsf{\exists}}}
\newcommand{\ieq}{\mathrel{\heavy\approx}}
\newcommand{\ineq}{\mathrel{\heavy{\not\approx}}}
\newcommand{\ichoice}{\heavy{\varepsilon}}
\newcommand{\idef}{\heavy{\rotiota}}
\newcommand{\ibigwedge}{\heavy{\bigwedge}}
\newcommand{\ibigvee}{\heavy{\bigvee}}

\newcommand{\iforallconstlight}{\light{\mathsf{\forall}}}
\newcommand{\iexistsconstlight}{\light{\mathsf{\exists}}}
\newcommand{\ieqlight}{\mathrel{\light\approx}}
\newcommand{\iandlight}{\mathbin{\light\land}}
\newcommand{\ineqlight}{\mathrel{\light\not\approx}}
\newcommand{\inotlight}{\light\neg}

\newcommand\iform[1]{[#1]}

\let\oldSigma=\Sigma
\def\Sigma{\mathrm{\oldSigma}}
\let\oldPi=\Pi
\def\Pi{\mathrm{\oldPi}}

\newcommand\semiileft{\llbracket}
\newcommand\semiiright[1]{\rrbracket_{#1}}
\newcommand\semleft{\semiileft}
\newcommand\semright[2]{\semiiright{#1\smash{,#2}}}
\newcommand\semii[2]{{\semiileft{#1}\semiiright{#2}}}
\newcommand\sem[3]{\semii{#1}{#2\smash{,#3}}}

\newcommand\medrightarrow{\mathrel{{{\color{black}\relbar}\kern-0.9ex\rlap{\color{white}\ensuremath{\blacksquare}}\kern-0.9ex}\joinrel{\color{black}\rightarrow}}}
\newcommand\medleftarrow{\mathrel{{\color{black}\leftarrow}\kern-0.9ex\rlap{\color{white}\ensuremath{\blacksquare}}\kern-0.9ex\joinrel{{\color{black}\relbar}}}}
\newcommand\medleftrightarrow{\mathrel{\leftarrow\kern-1.685ex\rightarrow}}
\newcommand\Medrightarrow{\mathrel{{{\color{black}\Relbar}\kern-0.9ex\rlap{\color{white}\ensuremath{\blacksquare}}\kern-0.9ex}\joinrel{\color{black}\Rightarrow}}}
\newcommand\Medleftrightarrow{\mathrel{\Leftarrow\kern-1.685ex\Rightarrow}}

\newcommand\omicron{o}
\newcommand{\HL}{\ensuremath{\mathrm{HL}}}
\newcommand{\PHL}{\ensuremath{\mathrm{HL_\mathrm{p}}}}
\newcommand{\cst}[1]{{\mathsf{#1}}}
\newcommand\tyargs[1]{\langle#1\rangle}
\newcommand\UNIF{\mathrel{\smash{\stackrel{\lower.1ex\hbox{\ensuremath{\scriptscriptstyle ?}}}{=}}}}
\newcommand{\eq}{\mathbin{\approx}}
\newcommand{\noteq}{\mathbin{\not\approx}}
\newcommand{\seq}[1]{\ensuremath{(#1)_{n=1}^{\infty}}}
\newcommand{\impk}[4]{\ensuremath{#1 \impl^{#3}_{#4} #2}}
\newcommand{\infname}[1]{\textsc{#1}}
\newcommand{\namedsimp}[3]{\prftree[d][r]{\relax{\infname{#1}}}{\strut#2}{\strut#3}}
\newcommand{\namedsimpsc}[4]{\prftree[d][r]{\relax{\infname{#1}\quad\text{if~}$#4$}}{\strut#2}{\strut#3}}
\newcommand{\impl}{\ensuremath{\medrightarrow}}
\newcommand{\bigpimpl}[1]{\hookrightarrow_\mathrm{p}}
\newcommand{\bigfimpl}[1]{\hookrightarrow}
\newcommand{\eqlit}[2]{\ensuremath{#1 \eq #2}}
\newcommand{\neqlit}[2]{\ensuremath{#1 \noteq #2}}
\newcommand{\pospredlit}[1]{{#1}}
\newcommand{\negpredlit}[1]{\neglit{#1}}
\newcommand{\neglit}[1]{\neg{\>#1}}
\newcommand{\arbpredlit}[1]{(\neg)\>\pospredlit{#1}}
\renewcommand{\cor}{\ensuremath{\mathrel{\lor}}}
\newcommand{\tupleii}[2]{\ensuremath{\vv{#1}{\!}_{#2}}}
\newcommand{\tuple}[1]{\ensuremath{\vv{#1}}}
\newcommand{\flatres}{\mathbin{\smash{\hbox{\Large$\kern-.1ex\rtimes$}}}}
\newcommand{\flatresiter}{\mathbin\leadsto}
\newcommand{\withoutpred}[2]{\,\overline{\!#1}_{\!#2}}
\newcommand{\withpred}[2]{#1_{#2}}
\newcommand{\withpredpos}[2]{#1^+_{#2}}
\newcommand{\withpredneg}[2]{#1^-_{#2}}
\newcommand{\binset}[1]{#1{\downarrow}_2}
\newcommand\OK[1]{#1}
\newcommand\NOK[1]{{\color{red}\textbf{#1}}}
\newcolumntype{d}{D{.}{.}{0}}
\newcommand{\confrep}[2]{\begin{conf}#1\end{conf}\begin{rep}#2\end{rep}}

\definecolor{light-gray}{gray}{0.925}
\newcommand\MAX[1]{\smash{\setlength{\fboxsep}{.3ex}\colorbox{light-gray}{\ensuremath{\vphantom{('q}{#1}}}}}

\maketitle

\newcommand{\Sigmaty}{\Sigma_\mathsf{ty}}
\newcommand{\LL}{\mathscr{L}}
\newcommand{\UU}{\mathscr{U}}
\newcommand{\VV}{\mathscr{V}}
\newcommand{\Vty}{\VV_\mathsf{ty}}
\newcommand{\III}{\mathscr{I}}
\newcommand{\II}{\mathscr{J}}
\newcommand{\IIty}{\II_\mathsf{ty}}
\newcommand{\IIIty}{\III_\mathsf{ty}}
\newcommand{\DD}{\mathscr{D}}

\begin{sloppypar}
\begin{abstract}
We generalize several propositional preprocessing techniques to higher-order
logic, building on existing first-order generalizations. These techniques
eliminate literals, clauses, or predicate symbols from the problem, with the aim
of making it more amenable to automatic proof search. We also introduce a new
technique, which we call \emph{quasipure literal elimination}, that strictly
subsumes pure literal elimination. The new techniques are implemented in the
Zipperposition theorem prover. Our evaluation shows that they sometimes help
prove problems originating from Isabelle formalizations and the TPTP library.
\end{abstract}
\end{sloppypar}

\section{Introduction}
\label{sec:introduction}

Processing techniques are an important optimization in SAT (Boolean
satisfiability) solving. Following up on early work in the 1990s
\cite{go-1992-so-pred-elim,hjo-1996-scan}, there has recently been renewed
interest in adapting propositional techniques to first-order logic
\cite{kk-2016-pe-fol,ksstb-2017-blockedfol,vukmirovic-et-al-2023-sat}, resulting
in a noticeable increase of the success rate of automatic theorem provers
based on superposition \cite{bachmair-ganzinger-1994}.

In this \paper, we consider the extension of four main classes of SAT
preprocessing techniques to classical higher-order logic. These extensions are
called
hidden-literal-based elimination (Section~\ref{sec:hidden-literal-based-elimination}),
predicate elimination (Section~\ref{sec:predicate-elimination}),
blocked clause elimination (Section~\ref{sec:blocked-clause-elimination}),
and quasipure literal elimination (Section~\ref{sec:quasipure-literal-elimination}).
Elimination techniques make the problem simpler and hence possibly more
amenable to automatic proof search. The techniques can be used either to
preprocess the problem or to transform the prover's current clause set during
proof search, a use that is sometimes called \emph{inprocessing}.

An advantage of preprocessing is its greater generality: Preprocessing
techniques can be used in tandem with any higher-order proof calculus, as long
as the calculus is built around a notion of clause. We assume that a clausifier
is run as a preprocessor and introduces some clausal structure. The more clausal
structure it produces, the more effective the elimination techniques can be.
Examples of provers compatible with the techniques are $\lambda$E
\cite{vukmirovic-et-al-2023-lambdae}, Leo-III \cite{steen-benzmueller-2018},
Vampire \cite{bhayat-reger-2020-combsup}, and Zipperposition
\cite{bentkamp-et-al-2021-hosup}.

Our setting is a clausal version of classical rank-1 polymorphic higher-order
logic (Section~\ref{sec:clausal-higher-order-logic}). Since previous work focused
on an untyped or monomorphic logic, our work also generalizes this aspect.
Higher-order logic also distinguishes between standard and general (Henkin)
semantics. Since calculi are proved complete with respect to general semantics,
this is the semantics we adopt. Our techniques preserve the unsatisfiability of
problems, and therefore their provability by a complete prover. In addition,
they preserve the satisfiability of problems, and therefore their unprovability
by a sound prover.

The main difficulty we face in higher-order logic concerns predicate elimination
and blocked clause elimination, which are both based on resolution. In
first-order logic, a literal $\lnot\>\cst{p}(\tuple{s})$ can only be resolved
against a literal $\cst{p}(\tuple{t})$, with the same predicate symbol
$\cst{p}$. By contrast, in higher-order logic, $\lnot\>\cst{p}\>\tuple{s}$ can
be resolved against any variable-headed literal $y\,\tuple{t}$, for example by
taking $y := \lambda\tuple{x}.\>\cst{p}\>\tuple{s}$, where the bound variables
$\tuple{x}$ are fresh. We will see that this issue can be circumvented: A key
finding of this \paper{} is that we can ignore variable-headed resolvents and
focus on the $\cst{p}$-literals.

Another potential
issue is that higher-order logic can have infinitely many resolvents. For
example, resolving $\lnot\>\cst{p}\>(\cst{f}\>(y\>\cst{a}))$ and
$\cst{p}\>(y\>(\cst{f}\>\cst{a}))$ produces infinitely many conclusions of the
form $\cst{p}\>(\cst{f}\,(\ldots(\cst{f}\>\cst{a})\ldots))$. Again, the issue is
not as severe as it looks, because the variant of resolution we use---flat
resolution---does not unify terms.

\begin{exa}
\label{ex:blocked-choice}
To give a flavor of our elimination techniques, let us review an example
involving blocked clause elimination.
Let $\cst{a} : \iota$, $\cst{p} : \iota \to \omicron$, and
$\cst{choice} : (\iota \to \omicron) \to \iota$ be symbols,
where $\omicron$ is the type of Booleans and $\iota$ is a base type.
Consider the clause set
\[N =
\{\lnot\>y\>z \cor y\>(\cst{choice}\> y){,}\;\,
\cst{q}\>\cst{a}{,}\;\,
\lnot\>\cst{q}\>(\cst{choice}\> \cst{q}){,}\;\,
\cst{p}\>\cst{a}{,}\;\,
\lnot\> \cst{p}\>z \cor z \eq \cst{a}\}
\]
The set is clearly inconsistent because the $\{y \mapsto \cst{q}{,}\; z \mapsto
\cst{a}\}$ instance of the first clause is inconsistent with the second and
third clauses.

Under some basic conditions, a clause $C$ containing a literal $L$ is said to be
``blocked'' if all of its so-called binary flat $L$-resolvents with clauses
from $N \setminus \{C\}$ are tautologies. (We will see in
Section~\ref{sec:blocked-clause-elimination} what this means exactly.)
The fourth clause is blocked by its literal $\cst{p}\>\cst{a}$ because its only
binary flat $(\cst{p}\>\cst{a})$-resolvent, with the fifth clause,
is the tautology $\cst{a} \noteq z \cor z \eq \cst{a}$.
Similarly, the fifth clause is blocked by its literal $\lnot\> \cst{p}\>z$
because its only binary flat $(\lnot\> \cst{p}\>z)$-resolvent, with the fourth
clause, is the same tautology. Either or both clauses can be removed without
from $N$ losing unsatisfiability.
\end{exa}

All the techniques are implemented in the higher-order prover Zipperposition
(Section~\ref{sec:implementation}), allowing us to measure their effectiveness on
benchmarks originating from Isabelle
\cite{nipkow-et-al-2002} formalizations
and the TPTP library \cite{sutcliffe-2017-tptp} (Section~\ref{sec:evaluation}).
The raw experimental data are available online.\footnote{\url{https://zenodo.org/record/7448169}}

We will mention closely related research in the relevant sections. We point to
Vukmirovi\'c et al.~\cite[Section~8]{vukmirovic-et-al-2023-sat} for a more detailed
discussion of related work.

\section{Clausal Higher-Order Logic}
\label{sec:clausal-higher-order-logic}

The logic we use as a basis of our work is a rank-1 polymorphic higher-order
logic with general semantics and both functional and Boolean
extensionality. It corresponds essentially to the logic embodied by the TPTP TH1
format \cite{kaliszyk-et-al-2016}, including Hilbert choice.
Our conventions largely follow those used by Bentkamp, Blanchette, Tourret,
and Vukmirovi\'c to define $\lambda$-superposition
\cite{bentkamp-et-al-2021-hosup}. Our presentation is based on theirs.

In higher-order logic, formulas are simply terms of Boolean type. Briefly, our
version of the logic also has a clausal outer structure, as found in several
higher-order provers. Clauses are then built around terms as an extra layer of
structure. We write formula-level Boolean operators in bold (e.g., $\inot$,
$\ior$, $\iforall$) to distinguish them from clause-level operators.

\oursubsection{Syntax}

Let us define the syntax of our logic more precisely, starting with types.
Throughout this \paper, we use the notation $\tupleii{a}{n}$, or simply
$\tuple{a}$, to denote an $n$-tuple $(a_1, \dotsc, a_n)$. Sometimes we might
also write $\tuple{a_i}$, meaning $(a_{i1}, \dotsc, a_{in})$, to be
distinguished from $\tupleii{a}{i} = (a_1, \dotsc, a_i)$.

We start by fixing an infinite set $\Vty$ of type variables. A set $\Sigmaty$ of
type constructors with associated arities is a \emph{type signature}. We require
the presence of a nullary type constructor $\omicron$, for Booleans, and a
binary type constructor $\to$ for functions. We let $\alpha$
range over type variables and $\kappa$ over type constructors.
A \emph{type}, ranged over by $\tau$ and $\upsilon$, is defined inductively to
be either a variable $\alpha \in \Vty$ or an expression
$\kappa(\tupleii{\tau}{n})$, where $\kappa$ is an $n$-ary type constructor
and $\tuple{\tau}$ in an $n$-tuple of types. If $n = 0$, we write $\kappa$
instead of $\kappa()$. In addition, expressions with $\kappa = {\to}$ are written
in infix notation, as $\tau_1 \to \tau_2$.
A \emph{type declaration} is an expression $\Pi\tupleii{\alpha}{m}.\>\tau$, where
$\tuple{\alpha}$ consists of distinct type variables and
all the type variables occurring in $\tau$ belong to $\tuple{\alpha}$. If $m = 0$,
we write $\tau$ instead of $\Pi.\>\tau$.

Next, we fix a type signature $\Sigmaty$ and a set $\VV$ of term variables with
associated types. We require that there are infinitely many variables of each
type. A \emph{term signature} is a set $\Sigma$ of \emph{symbols} $\cst{a},
\cst{b}, \cst{c}, \cst{f}, \cst{p}, \cst{q}, \dotsc,$ each associated with a
type declaration. (Often, symbols are called ``constants'' in the higher-order
logic literature.)
We write $\cst{f} : \Pi\tuple{\alpha}.\>\tau$ to indicate
that symbol $\cst{f}$ has type signature $\Pi\tuple{\alpha}.\>\tau$.
We require the presence in $\Sigma$ of the logical symbols
$\hbox{$\ifalse$},\allowbreak \hbox{$\itrue$},\allowbreak
\hbox{$\inot$},\allowbreak \hbox{$\iand$},\allowbreak \hbox{$\ior$},\allowbreak
\hbox{$\iimplies$},\allowbreak \hbox{$\iforall$},\allowbreak
\hbox{$\iexists$},\allowbreak \hbox{$\ieq$},\allowbreak \hbox{$\ineq$}$
with their usual type declarations (e.g., $\Pi\alpha.\>\alpha \to \alpha
\to \omicron$ for $\ieq$). We also assume the presence of the Hilbert choice
operator $\hbox{\ichoice} : \Pi\alpha.\>(\alpha\to\omicron)\to\alpha$.
We will generally leave the signature implicit, assuming some fixed signature
$\Sigma$.

Polymorphic higher-order logic does not distinguish between function and
predicate symbols. Nonetheless, it will be convenient to refer to symbols that
can yield Boolean values as predicate symbols. Specifically, we will call
a symbol with type declaration $\cst{p} : \Pi\tuple{\alpha}.\>\tau$ that admits
an instance of the form $\cst{p}\tyargs{\tuple{\tau}} : \>\upsilon_1 \to \cdots \to
\upsilon_n \to \omicron$ a \emph{predicate symbol}.

We now introduce three notions of terms:\ raw $\lambda$-terms, $\lambda$-terms, and
actual terms. The $\lambda$-terms are $\alpha$-equivalence classes of raw
$\lambda$-terms, and the actual terms are $\beta\eta$-equivalence classes of
$\lambda$-terms.

More precisely, the \emph{raw $\lambda$-terms} are defined inductively as follows:
\begin{itemize}
\item Every variable $x$ of type $\tau$ is a raw $\lambda$-term of type $\tau$.

\smallskip

\item If $\cst{f}$ has type declaration $\Pi\tuple{\alpha}_m.\>\tau$ in $\Sigma$ and
  $\tuple{\upsilon}_m$ is a tuple of types, called \emph{type arguments}, then
  $\cst{f}\tyargs{\tupleii{\upsilon}{m}}$ is a
  raw $\lambda$-term of type $\tau\{\tupleii{\alpha}{m} \mapsto \tupleii{\upsilon}{m}\}$.
  If $m = 0$, we simply write $\cst{f}$ instead of $\cst{f}\tyargs{}$.

\smallskip

\item If $x$ is a variable of type $\tau$ and $t$ is a term of type $\upsilon$,
  then the \emph{$\lambda$-abstraction}
  $\lambda x.\> t$ is a raw $\lambda$-term of type $\tau\to\upsilon$.

\smallskip

\item If $s$ is a term of type $\tau\to\upsilon$ and $t$ is a term of type
  $\tau$, then the \emph{application} $s\>t$ is a raw $\lambda$-term of type
  $\upsilon$.
\end{itemize}

We abbreviate $\lambda x_1.\,\dots\,\lambda x_n.\> t$ to
$\lambda x_1\,\dots\,x_n.\> t$ or $\lambda\tupleii{x}{n}.\> t$,
$\iforall\>(\lambda x.\>t)$ to $\iforall x.\>t$, and similarly
for $\iexists$.
Abusing notation, we also write $t\>\tupleii{u}{n}$ for $t\>u_1\,\dots\,u_n$.
We assume standard notions of free and bound variables as well as subterms.
To indicate that a term~$t$ has a type~$\tau$, we write $t : \tau$.

\looseness=-1
The $\alpha$-renaming rule of the $\lambda$-calculus relates two raw
$\lambda$-terms if the two are equal up to (capture-avoiding) renaming of their
bound variables. For example, $\lambda x.\>\cst{f}\>x\>x$ and
$\lambda y.\>\cst{f}\>y\>y$ are $\alpha$-renamings of each other.
Two raw $\lambda$-terms are $\alpha$-equiv\-alent if they can be made equal by
$\alpha$-renaming their subterms.
The \emph{$\lambda$-terms} consist of the equivalence classes of raw
$\lambda$-terms modulo $\alpha$-equivalence of subterms. We assume the standard
notion of (capture-avoiding) substitution on $\lambda$-terms. We also
define a notion of replacement: $t[\cst{f} \mapsto u]$ denotes the term obtained
by replacing all occurrences of $\cst{f}$ in $t$ with a term $u$ of the same
type.

The $\beta$-reduction rule relates two $\lambda$-terms
if the first one has the form $(\lambda x.\>s)\>t$ and the second one has
the form $s\{x \mapsto t\}$, where bound variables in $s$ are implicitly renamed to
avoid capture. The $\eta$-reduction rule relates two $\lambda$-terms
if the first one has the form $\lambda x.\>t\>x$ and the second one has the
form $t$, where $t$ contains no free occurrences of $x$. For example,
$(\lambda x.\>\cst{f}\>x\>x)\>\cst{b}$ $\beta$-reduces to
$\cst{f}\>\cst{b}\>\cst{b}$, and $\lambda x.\>\cst{f}\>x$ $\eta$-reduces to
$\cst{f}$. Two $\lambda$-terms are $\beta\eta$-equivalent if they can be made
equal by $\beta$- and $\eta$-reducing their subterms.
The \emph{terms} consist of the equivalence classes of $\lambda$-terms modulo
$\beta\eta$-equivalence of subterms.
The \emph{formulas} are the terms of type $\omicron$.
We let $\varphi, \psi$ range over formulas.

\begin{conventionx}
When inspecting the structure of a term, we will consider a representative in
$\eta$-short $\beta$-normal form, obtained by exhaustively applying $\beta$- and
$\eta$-reduction on subterms. Such a representative is unique up to
$\alpha$-equivalence.
\end{conventionx}

An alternative to the $\eta$-short $\beta$-normal form is the $\eta$-long
$\beta$-normal form, in which unapplied functions are $\eta$-expanded rather
than $\eta$-reduced (i.e., $\eta$-reduction is applied in reverse on these). The
techniques presented in this \paper{} work unchanged in such a setting.

Two terms $t, u$ are \emph{unifiable} if there exists a substitution $\sigma$
such that $t\sigma = u\sigma$. For example, $\cst{a}$ and $y\>\cst{a}$ are
unifiable by taking $y := \lambda x.\>x$. (This works because terms are equal up
to $\beta$-reduction.) Unification of higher-order terms is undecidable. Types,
however, are isomorphic to first-order terms, and hence their unification
problem is decidable. Moreover, if a unifier exists, then a most general unifier
exists (up to the naming of variables). For example, the most general unifier
for the unification problem $\cst{pair}(\alpha, \cst{nat}) \UNIF
\cst{pair}(\cst{int}, \beta)$ is $\{\alpha \mapsto \cst{int}{,}\; \beta \mapsto
\cst{nat}\}$.

Finally, we define literals and clauses on top of terms.
An atom is an equation $\eqlit{s}{t}$ corresponding to an unordered pair $\{s,
t\}$. (We reserve $=$ for syntactic equality of terms.)
A literal is an equation $\eqlit{s}{t}$ or a disequation $\neqlit{s}{t}$.
Given a predicate symbol~$\cst{p}$, the literal
$\eqlit{\cst{p}\tyargs{\tuple{\alpha}}\>\tuple{s}}{\itrue}$
is abbreviated to $\cst{p}\tyargs{\tuple{\alpha}}\>\tuple{s}$, and its
complement
$\neqlit{\cst{p}\tyargs{\tuple{\alpha}}\>\tuple{s}}{\itrue}$ is abbreviated to
$\neglit{\cst{p}\tyargs{\tuple{\alpha}}\>\tuple{s}}$.
Moreover, a \emph{$\cst{p}$-literal} is a literal of
the form $(\lnot)\,\cst{p}\tyargs{\tuple{\tau}}\>\tuple{t}$. Note that it is
possible in higher-order logic for a non-$\cst{p}$-literal to contain $\cst{p}$,
or even for the arguments $\tuple{t}$ of a $\cst{p}$-literal to contain
$\cst{p}$.
Given a literal $L$, we write $\neglit{L}$ to denote its complement, with
$\neglit{\neglit{L}} = L$.

A clause $C$ is a finite multiset of literals, written as $L_1 \cor \cdots \cor
L_n$ and interpreted disjunctively. Clauses are often defined as sets of
literals, but multisets are better behaved with respect to substitution: If $C$
has $n$~literals, so has $C\sigma$ regardless of whether $\sigma$ unifies some
of $C$'s literals. The type and term variables contained in a clause are
implicitly quantified universally. (Within terms, $\iforall$ and $\iexists$ can
be used to quantify over term variables.)

A type, term, or clause is \emph{monomorphic} if it contains no type variables.
A term or clause is \emph{closed} if it contains no free term variables.

It is sometimes useful to encode a clause $C$ as a formula. The formula
$\iform{C}$ representing the clause $C$ is defined by replacing the nonbold
symbols $\eq$, $\noteq$, and $\lor$ by their bold counterparts $\ieq$, $\ineq$,
and $\ior$. This formula is uniquely defined up to the orientation of the
equations and the order of the literals, neither of which affects the semantics.

\oursubsection{Semantics}

A \emph{type interpretation} $\IIIty = (\UU, \IIty)$ consists of two components.
The \emph{universe} $\UU$ is a collection of nonempty sets, the
\emph{domains}. We require $\UU$ to contain the domain $\{0,1\}$, where 
0 represents falsehood and 1 represents truth.
The function $\IIty$ associates
with each $n$-ary type constructor~$\kappa$
a function
$\IIty(\kappa) : \UU^n \rightarrow \UU$,
with the requirements that
$\IIty(\omicron) = \{0,1\}$
and that the set $\IIty(\to\nobreak)(\DD_1,\allowbreak\DD_2)$
is a subset of the (total) function space from $\DD_1$ to $\DD_2$
for all domains $\DD_1,\DD_2\in\UU$.
The semantics is \emph{standard} if
$\IIty({\to})(\DD_1,\DD_2)$ is the entire function space for all $\DD_1,\DD_2\in\UU$.
A \emph{type valuation} $\xi$ is a function that maps every type variable to a domain.

The \emph{denotation} of a type in a type interpretation $\IIIty$ under a type
valuation $\xi$ is defined by the recursive equations
$\sem{\alpha}{\IIIty}{\xi}=\xi(\alpha)$ and
$\sem{\kappa(\tuple{\tau})}{\IIIty}{\xi}=
\IIty(\kappa)(\sem{\tuple{\tau}}{\IIIty}{\xi})$.
For monomorphic types $\tau$, the denotation does not depend on the valuation
$\xi$, allowing us to write $\semii{\tau}{\IIIty}$ instead of
$\sem{\tau}{\IIIty}{\xi}$.

A type valuation $\xi$ can be extended to be a \emph{valuation} by additionally
assigning an element $\xi(x)\in\sem{\tau}{\IIIty}{\xi}$ to each variable $x : \tau$.
We will sometimes use partial functions as valuations if the values outside the
function's domain are irrelevant.
An \emph{interpretation function} $\II$ for a type interpretation $\IIIty$ associates with each symbol
$\cst{f} : \Pi{\tupleii{\alpha}{m}}.\>\tau$ and domain tuple $\tuple{\DD}\in\UU^m$
a value
$\II(\cst{f},\tuple{\DD}) \in
\sem{\tau}{\IIIty}{[\tuple{\alpha}\mapsto\tuple{\DD}]}$,
We require that the logical symbols are interpreted in the usual way in terms of
0 and 1. For example,
$\II(\ifalse) = 0$ and $\II(\ior)(v,w) = \max\,\{v,w\}$.
Note that in the presence of $\ichoice$ in the signature,
every type $\tau$ must be interpreted by a nonempty set, for
$\II(\ichoice\tyargs{\tau}\>(\lambda x.\>\ifalse))$ to be defined.

We require the comprehension principle to hold: Every function designated
by a $\lambda$-abstraction is contained in the domain associated with its type.
We initially allow $\lambda$-abstractions to designate arbitrary elements of the
domain. This enables us to define the denotation of a term. Then we impose
restrictions to rule out undesirable $\lambda$-abstraction designations.

A \emph{$\lambda$-designation function} $\LL$
for a type interpretation $\IIIty$ is a function that maps
a valuation $\xi$ and a $\lambda$-abstraction of type $\tau$ to an element
of $\sem{\tau}{\IIIty}{\xi}$.
An \emph{interpretation} $\III = (\IIIty,\II,\LL)$ combines
a type interpretation, an interpretation function, and a $\lambda$-designation
function.

For an interpretation~$\III$ and a valuation~$\xi$, the \relax{denotation of a term} is defined
recursively as
$\sem{x}{\III}{\xi} = \xi(x)$,
$\sem{\cst{f}\tyargs{\tuple{\tau}}}{\III}{\xi} =
  \II(\cst{f},\sem{\tuple{\tau}}{\IIIty}{\xi})$,
$\sem{s\>t}{\III}{\xi} = \sem{s}{\III}{\xi} (\sem{t}{\III}{\xi})$, and
$\sem{\lambda x.\> t}{\III}{\xi} = \LL(\xi,\lambda x.\> t)$.
For monomorphic closed terms $t$, the
denotation does not depend on the valuation $\xi$,
allowing us to write $\semii{t}{\III}$ instead of $\sem{t}{\III}{\xi}$.

An interpretation $\III$ is \emph{proper} if $\sem{\lambda
x \mathbin: \upsilon.\>t}{\III}{\xi}(v) = \sem{t}{\III}{\xi[x\mapsto v]}$ for every
$\lambda$-ab\-straction $\lambda x \mathbin: \upsilon.\>t$, every valuation~$\xi$,
and every value $v \in \sem{\upsilon}{\IIIty}{\xi}$.
We will assume throughout that all interpretations are proper and will
construct only proper interpretations.
If a type interpretation
$\IIIty$ and an interpretation function $\II$ can be extended by a $\lambda$-designation function $\LL$ to an
interpretation $(\IIIty,\II,\LL)$, then this $\LL$ is unique
\cite[Proposition~2.18]{fitting-2002}.

Given an interpretation $\III$ and a valuation $\xi$, an equation $s\eq t$ is
\relax{true} if
$\sem{s}{\III}{\xi}$ and $\sem{t}{\III}{\xi}$ are equal
and it is \relax{false} otherwise.
A disequation $s\noteq t$ is true if $s \eq t$ is false.
A clause is \relax{true} if at least one of its literals is true.
A clause set is \relax{true} if all the clauses it contains are true.
An interpretation $\III$ is a \emph{model} of a clause set $N$,
written $\III \models N$, if $N$ is true in $\III$ for every valuation~$\xi$.

A clause $C$ is a \emph{tautology} if $\III \models \{C\}$ for every
interpretation $\III$.
It is \emph{satisfiable} if there exists an interpretation $\III$
such that $\III \models \{C\}$; otherwise, it is \emph{unsatisfiable}.
Notice that these concepts are defined with respect to general, and not
necessarily standard, interpretations.

It is sometimes convenient to assume that the interpretation of
monomorphic types is injective on types---that is, for all
monomorphic $\tau, \upsilon$, if
$\semii{\tau}{\IIIty} = \semii{\upsilon}{\IIIty}$, then $\tau = \upsilon$.
This assumption is reasonable because the elements of $\IIty(\kappa)$,
where $\kappa \notin \{\omicron,\to\}$, are immaterial
and can be renamed if desired and because the set-theoretic representation of
functions, as nonempty sets of pairs, preserves this property.
(The sets of pairs are nonempty thanks to the presence of $\ichoice$ in the
signature, as noted above.) We call this principle the \emph{distinct domain
assumption}.

\section{Hidden-Literal-Based Elimination}
\label{sec:hidden-literal-based-elimination}

In propositional logic \cite{hjb-2011-big-simplification} and clausal first-order
logic \cite{vukmirovic-et-al-2023-sat}, a hidden literal for a literal $L$ and a
clause set $N$ is a literal that can be added or removed from any clause
containing $L$ without affecting its truth value in models of $N$. Several
elimination techniques are based on hidden literals; in particular, hidden
literal elimination removes hidden literals from clauses in which they
occur.

\begin{exa}
\label{ex:hidden-lit}
Consider the literal $\cst{c}$ and the clause set $N = \{\lnot\>\cst{a} \lor
\cst{b}{,}\; \lnot\>\cst{b} \lor \cst{c}\}$. Then $\cst{b}$ is a hidden
literal: Since $\cst{b}$ implies $\cst{c}$ according to $N$, we have
that $\cst{b} \lor \cst{c}$ and $\cst{c}$ have the same truth value in models of
$N$. Similarly, since $\cst{a}$ implies $\cst{c}$ (by transitivity),
$\cst{a}$ is also a hidden literal. Thus, hidden literal elimination would
reduce the clause $\cst{a} \lor \cst{b} \lor \cst{c}$ to~$\cst{c}$.
\end{exa}

The first-order definitions of hidden literals, hidden tautologies, hidden
literal elimination, hidden tautology elimination, failed literal elimination,
hidden tautology reduction, and failed literal reduction
\cite{vukmirovic-et-al-2023-sat} work verbatim in clausal higher-order logic,
for both preprocessing and inprocessing.
All the techniques preserve satisfiability and unsatisfiability.
We call them collectively \emph{hidden-literal-based elimination} (HLBE).

\section{Predicate Elimination}
\label{sec:predicate-elimination}

\emph{Predicate elimination} (PE)
\cite{go-1992-so-pred-elim,kk-2016-pe-fol,vukmirovic-et-al-2023-sat}
is a set of techniques that remove all occurrences of some
predicate symbol in a first-order problem by resolving clauses that contain
it. Predicate elimination generalizes variable
elimination in propositional logic \cite{sp-04-niver,cs-00-zres}.
In this section, we generalize two specific techniques to higher-order logic.

\oursubsection{Singular Predicate Elimination}

The first technique is called singular (or ``non-self-referential'') predicate
elimination. Definitions \ref{def:occurs-deep}~to~\ref{def:flat-res-set} below
are adapted from monomorphic first-order logic.

\begin{defi}
\label{def:occurs-deep}
A predicate symbol $\cst{p}$ \emph{occurs deep} in a clause $C$
if it occurs in a position other than as the head of the atom of a
$\cst{p}$-literal somewhere in $C$. The symbol $\cst{p}$ \emph{occurs deep} in
a clause set $N$ if it occurs deep in one of its clauses $C \in N$.
\end{defi}

\begin{defi}
A predicate symbol $\cst{p}$ is called \emph{singular} for a
clause $C$ if these conditions are met:
\begin{enumerate}
\item $C$ contains at most one $\cst{p}$-literal;
\item\label{itm:singular-ii} $\cst{p}$ does not occur deep in $C$.
\end{enumerate}
The symbol $\cst{p}$ is \emph{singular} for a clause set $N$ if
$\cst{p}$ is singular for every clause contained in $N$.
\end{defi}

\begin{sloppypar}
\begin{defi}
\label{def:polymorphism-safe}
The clause $C = (\lnot)\,\cst{p}\tyargs{\tuple{\tau}}\>\tuple{t} \lor C'$ is
\emph{polymorphism-safe}
for its literal $(\lnot)\,\cst{p}\tyargs{\tuple{\tau}}\>\tuple{t}$
if all type variables occurring in $C$ occur in $\tuple{\tau}$.
A clause $C$ is \emph{polymorphism-safe} for $\cst{p}$ if it is polymorphism-safe
for all its $\cst{p}$-literals.
A clause set $N$ is \emph{polymorphism-safe} for $\cst{p}$ if all the clauses
it contains are polymorphism-safe for $\cst{p}$.
\end{defi}
\end{sloppypar}

\begin{defi}
\label{def:flat-resolvent}
Let $C = \cst{p}\tyargs{\tuple{\tau}}\>\tupleii{s}{n} \lor C'$ and
$D = \lnot\>\cst{p}\tyargs{\tuple{\upsilon}}\>\tupleii{t}{n}
\lor D'$. The \emph{flat resolvent} of $C$ and $D$ on
$\cst{p}\tyargs{\tuple{\tau}}\>\tupleii{s}{n}$ and
$\lnot\>\cst{p}\tyargs{\tuple{\upsilon}}\>\tupleii{t}{n}$ is
the clause $(s_1 \noteq t_1 \cor \cdots
\cor s_n \noteq t_n \cor C' \cor D')\sigma$,
where $\sigma$ is the most general unifier of
$\tuple{\tau} \UNIF \tuple{\upsilon}$. The flat resolvent is not defined if
$\tuple{\tau}$ and $\tuple{\upsilon}$ are not unifiable.
\end{defi}

Flat resolvents were already present in the first-order setting
\cite{go-1992-so-pred-elim,kk-2016-pe-fol,vukmirovic-et-al-2023-sat}. They
are reminiscent of Huet's approach to unification in higher-order resolution
\cite{huet-1973}, which is also used by Benzm\"uller and colleagues
\cite{benzmueller-kohlhase-1998,steen-benzmueller-2021}.

\begin{defi}
\label{def:flat-res-set}
Let $M, N$ be clause sets and $\cst{p}$ be a singular predicate for $M$.
Let $\flatresiter$ be the following relation on clause set pairs and clause
sets:
\begin{enumerate}
\item
\label{itm:flat-res-set-res}
$(M{,}\> \{\arbpredlit{\cst{p}\tyargs{\tuple{\tau}}\>\tuple{s}} \cor C'\} \uplus N) \flatresiter (M{,}\; N' \cup N)$ if
$N'$ is the set that consists of all clauses, up to variable renaming, that are flat resolvents
on $\arbpredlit{\cst{p}\tyargs{\tuple{\tau}}\>\tuple{s}}$ and
$\arbpredlit{\negpredlit{\cst{p}\tyargs{\tuple{\upsilon}}\>\tuple{t}}}$
of $\arbpredlit{\cst{p}\tyargs{\tuple{\tau}}\>\tuple{s}} \cor C'$ and a clause
$\arbpredlit{\negpredlit{\cst{p}\tyargs{\tuple{\upsilon}}\>\tuple{t}}} \cor D'$
from $M$ as premises. The premises' variables are renamed apart.

\smallskip
\item \label{itm:flatresiter-exit}
 $(M, N) \flatresiter N$ if $N$ contains no $\cst{p}$-literals.
\end{enumerate}
\noindent The \emph{resolved set} $M \flatres_{\!\cst{p}} N$ is the
$(\Sigma\setminus \{\cst{p}\})$-clause set $N'$ such that $(M,N)
\flatresiter^* N'$.
\end{defi}

For finite sets $M, N$, the resolved set $N'$ is reached in a finite number of
steps, and it is unique up to variable renaming. The argument is as for
first-order logic
\cite[Lemma~4.4]{vukmirovic-et-al-2023-sat}.
Note that the result may contain deep occurrences of $\cst{p}$ if the initial
set $N$ contains such occurrences.

\begin{defi}
Let $N$ be a clause set and $\cst{p}$ be a singular predicate symbol for
$N.$ Let $\withpredpos{N}{\cst{p}}$ consist of all clauses belonging to $N$ that contain a positive $\cst{p}$-literal,
let $\withpredneg{N}{\cst{p}}$ consist of all clauses belonging to $N$ that contain a negative $\cst{p}$-literal,
let $\withpred{N}{\cst{p}} = \withpredpos{N}{\cst{p}} \cup \withpredneg{N}{\cst{p}}$,
and let $\withoutpred{N}{\cst{p}} = N \setminus \withpred{N}{\cst{p}}$.
\end{defi}

\begin{defi}
\label{def:pred-elim}
Let $N$ be a finite clause set that is polymorphism-safe for $\cst{p}$ and
$\cst{p}$ be a singular predicate for $N.$
\emph{Singular predicate elimination} (SPE) of $\cst{p}$ in $N$ replaces $N$ by
the $(\Sigma \setminus \{\cst{p}\})$-clause set
$\withoutpred{N}{\cst{p}} \cup (\withpredpos{N}{\cst{p}} \flatres_{\!\cst{p}} \withpredneg{N}{\cst{p}})$.
\end{defi}

SPE preserves satisfiability: The only clauses added are flat resolvents, and
flat resolution is clearly sound. In first-order logic, SPE also preserves
unsatisfiability \cite[Theorem~1]{kk-2016-pe-fol}. With a small restriction on
polymorphism, this result extends to polymorphic higher-order logic. Also
note that deep occurrences of $\cst{p}$ are not possible in the result, because
of the requirement that $\cst{p}$ be a singular predicate for the input $N.$

\begin{exa}
Thanks to the use of flat resolvents, the unification work is left to the proof
calculus. This is convenient, because higher-order unification is undecidable
and hence, in general, could not be done exhaustively in a preprocessing
technique.

Consider the clause set $N =
\{\cst{p}\> z\> z \cor \cst{q}\> z{,}\allowbreak\;
  \lnot\> \cst{p}\> (\cst{f}\> (y\> \cst{a}))\> (y\> (\cst{f}\> \cst{a})){,}\allowbreak\;
  \lnot\> \cst{q}\>\cst{b}\}$.
Applying SPE to $\cst{p}$ transforms
$N$ into $N' = \{
  \cst{f}\> (y\> \cst{a}) \noteq z
  \cor y\> (\cst{f}\> \cst{a}) \noteq z
  \cor \cst{q}\> z{,}\;
  \lnot\> \cst{q}\>\cst{b}\}$.
It is then the calculus's task to enumerate values for
$z$ that solve the unification problem $z \UNIF \cst{f}\> (y\> \cst{a}) \UNIF
y\> (\cst{f}\> \cst{a})$. These values are
\[\cst{f}\>\cst{a}{,}~
  \cst{f}\> (\cst{f}\> \cst{a}){,}~
  \cst{f}\> (\cst{f}\> (\cst{f}\> \cst{a})){,}~
  \ldots\]

In $\lambda$-superposition \cite{bentkamp-et-al-2021-hosup}, this enumeration
would be the responsibility of the \textsc{ERes} inference rule. The
first clause in $N'$ could be simplified to
$y\> (\cst{f}\> \cst{a}) \noteq \cst{f}\> (y\> \cst{a})
\cor \cst{q}\> (\cst{f}\> (y\> \cst{a}))$, by eliminating $z$.
Then \textsc{ERes} would unify the two sides of the first literal, producing
one conclusion per unifier. Since there are infinitely many unifiers, this would
lead to infinitely many conclusions:
\[\cst{q}\> (\cst{f}\> \cst{a}){,}~
  \cst{q}\> (\cst{f}\> (\cst{f}\> \cst{a})){,}~
  \cst{q}\> (\cst{f}\> (\cst{f}\> (\cst{f}\> \cst{a}))){,}~
  \ldots\]
Using dovetailing, this infinite enumeration can be interwoven with other
inferences and other activities of the prover
\cite[Section~6]{vukmirovic-et-al-2023-hot}.
\end{exa}

\begin{exa}
\label{ex:spe-ab}
Consider the satisfiable clause set $N =
\{\cst{p}\> (\cst{f}\>z) \cor \cst{q}\> z{,}\allowbreak\;
  \lnot\> \cst{p}\> (\cst{f}\>\cst{a})\}$.
SPE of $\cst{p}$ transforms $N$ into the equally satisfiable set
$N' = \{\cst{f}\>\cst{a} \noteq \cst{f}\>z \cor \cst{q}\> z\}$.

Note that although $\cst{p}$ is absent from $N'$,
a predicate that meets its specification can be created based on the
first clause of $N$, in which the $\cst{p}$-literal is positive.
This predicate is
$\lambda x.\> \iexists y.\> x \ieq \cst{f}\>y \iand \inot\>\cst{q}\>y$.
In general, we would use all the clauses in which the $\cst{p}$-literal is
positive and ignore the other clauses. (Dually, we could have defined the
predicate in terms of the clauses in which the $\cst{p}$-literal is negative.)
This predicate makes $\cst{p}$ true only when it must be
true to satisfy the first clause---namely, when the clause's
non-$\cst{p}$-literal is false.

If we replace $\cst{p}$ with this $\lambda$-abstraction in $N$ and
$\beta$-reduce, we obtain the satisfiable set
$N'' = \{(\iexists y.\> \cst{f}\>z \ieq \cst{f}\>y \iand \inot\>\cst{q}\>y)
  \cor \cst{q}\> z{,}\allowbreak\;
  \lnot\> (\iexists y.\> \cst{f}\>\cst{a} \ieq \cst{f}\>y \iand \inot\>\cst{q}\>y)\}$.
Nothing essential is lost by eliminating $\cst{p}$---if
we need $\cst{p}$, we can use the $\lambda$-abstraction. This idea is the key to
the proof of Theorem~\ref{thm:spe-equisat} below.
\end{exa}

\begin{exa}
\label{ex:spe-ya}
Consider the clause set $N = \{\lnot\> y\> \cst{a}{,}\; \cst{p}\>
\cst{a}\}$. It is easy to see that the set is unsatisfiable, by taking $y :=
\cst{p}$. SPE of $\cst{p}$ transforms $N$ into the set $N' = \{\lnot\> y\>
\cst{a}\}$. Although $y\> \cst{a}$ is unifiable with the literal $\cst{p}\>\cst{a}$,
the first clause is left unchanged by SPE.
Like $N$, $N'$ is unsatisfiable. We cannot witness unsatisfiability by taking $y
:= \cst{p}$.
We can take $y := \lambda x.\; x \eq \cst{a}$, simulating
$\cst{p}$ on the input $\cst{a}$.
Alternatively, we can take $y := \lambda x.\> \itrue$.
Indeed, we could have taken either instantiation for $y$ to show that
$N$ is unsatisfiable, without exploiting the presence of $\cst{p}\>\cst{a}$.
\end{exa}

Example~\ref{ex:spe-ya} corroborates our choice of ignoring literals headed by a
variable in the definition of SPE, focusing instead on $\cst{p}$-literals. The
intuition is that often $\cst{p}$ is unnecessary to have unsatisfiability, and
when it is necessary it can be simulated by a $\lambda$-abstraction that does
not contain it.

\begin{thm}
\label{thm:spe-equisat}
Let $N$ be a finite clause set that is polymorphism-safe for $\cst{p}$ and
$\cst{p}$ be a singular predicate symbol for $N$.
Let $N'$ be the result of applying SPE of $\cst{p}$ to $N$. Then $N'$ is
satisfiable if and only if $N$ is satisfiable.
\end{thm}

\begin{proof}
The ``if'' direction follows immediately from the soundness of flat resolution.
For the other direction, our strategy is inspired by Khasidashvili and
Korovin \cite[Theorem~1]{kk-2016-pe-fol}.

Let $\III = ({\IIIty}, \II, \LL)$ be a model of $N' =
\withoutpred{N}{\cst{p}} \cup (\withpredpos{N}{\cst{p}} \flatres_{\!\cst{p}}
\withpredneg{N}{\cst{p}})$, with ${\IIIty} = (\UU, \IIty)$.
We assume without loss of generality that this model satisfies the distinct
domain assumption (Section~\ref{sec:clausal-higher-order-logic}).
We will define a new interpretation $\III'$ that extends
$\III$ with a semantics for $\cst{p}$
and show that $\III'$ is a model of $N = \smash{\withoutpred{N}{\cst{p}}
\cup \withpredpos{N}{\cst{p}} \cup \withpredneg{N}{\cst{p}}}$.
To achieve this, we will separately
show that (1)~$\III' \models \withoutpred{N}{\cst{p}}$;
(2)~$\III' \models \withpredpos{N}{\cst{p}}$; and
(3)~$\III' \models \withpredneg{N}{\cst{p}}$.

The interpretation $\III' = ({\IIIty}, \II', \LL')$ we construct is identical to
$\III$ except that $\II'$ is extended so that
$\II'(\cst{p}, \tuple{\DD})$ and $\LL'$ are defined as follows,
by mutual recursion.

We start with $\II'(\cst{p}, \tuple{\DD})$.
We construct a suitable interpretation for $\cst{p}$ as a curried $n$-ary
predicate. Let $\tuple{\tau'}$ be monomorphic types such that
$\semii{\tuple{\tau'}}{\IIIty} = \tuple{\DD}$. By the distinct domain
assumption, these types are unique if they exists. If no such types exists, let
$\II'(\cst{p}, \tuple{\DD})$ be the predicate $\semii{\lambda\tuple{x}.\>
\ifalse}{\III}$. This predicate is guaranteed to exist in the interpretation of
the type of $\cst{p}\tyargs{\tuple{\tau'}}$ thanks to the comprehension
principle (Section~\ref{sec:clausal-higher-order-logic}). The predicate plays
the role of a placeholder for impossible interpretations of $\cst{p}$; its value
is irrelevant.

In the case where the types $\tuple{\tau'}$ exist, we will construct a
right-hand side for $\II'(\cst{p}, \tuple{\DD})$ using the same idea as in
Example~\ref{ex:spe-ab}. To cope with polymorphism, we will filter out any
clauses whose $\cst{p}$'s type arguments cannot be instantiated to
$\tuple{\tau'}$ and instantiate the remaining clauses.

More precisely, let $M$ be the smallest set such that
for each clause $C = \cst{p}\tyargs{\tuple{\upsilon}}\>\tupleii{t}{n} \lor C'$
contained in $\withpred{N}{\cst{p}}$,
if there exists a substitution $\sigma$ such that
$\tuple{\upsilon}\sigma = \tuple{\tau'}$, then have $M$ contain $C\sigma$.
Notice that the polymorphism-safety hypothesis ensures that $C\sigma$
is monomorphic and uniquely specified. This is desirable because we want to
assign a unique right-hand side to $\II'(\cst{p}, \tuple{\DD})$.

We now define a term $u$ whose interpretation $\semii{u}{\III} =
\semii{u}{\III'}$ will give us the right-hand side.
Let $\tupleii{x}{n}$ be a tuple of fresh
variables. With each clause
$\cst{p}\tyargs{\tuple{\tau'}}\>\tupleii{t}{n} \lor C'$ contained in $M$
and whose free variables are $\tuple{y}$, associate the formula
\[\iexists\tuple{y}.\> x_1 \ieq t_1 \iand \cdots \iand x_n \ieq t_n \iand \inot\> \iform{C'}\]
(Recall from Section~\ref{sec:clausal-higher-order-logic} that $\iform{C'}$
denotes a formula representing the clause $C'$.)
Let $\varphi_1,\dots,\varphi_k$ be all such formulas,
and let $\varphi = \varphi_1 \ior \cdots \ior \varphi_k$. Then
we take $u = \lambda\tuple{x}.\>\varphi$.
This choice of $u$ will ensure that $\III'$ satisfies every clause in
$\withpredpos{N}{\cst{p}}$,
and thanks to the comprehension principle,
the predicate denoted by $u$ is guaranteed to exist in the interpretation
of $\smash{\cst{p}\tyargs{\tuple{\tau'}}}$'s type.

To finish the definition of $\III'$, we must specify $\LL'$.
Given a term $t$, let $t[u/\cst{p}]$ denote the variant of the term $t$
in which all occurrences of $\cst{p}$ are replaced by the term $u$ as defined
above (for suitable types $\tuple{\tau'}$). For all valuations
$\xi$ and $\lambda$-abstrac\-tions~$\lambda x \mathbin: \upsilon.\>t$, we define
$\LL'(\xi,\lambda x.\>t)$ as the function that maps each $v \in
\sem{\upsilon}{\IIIty}{\xi}$ to $\sem{t[u/\cst{p}]}{\III}{\xi[x\mapsto v]}$.
This function exists in the domain associated with the $\lambda$-abstraction's type
because $\III$ obeys the comprehension principle. Moreover, because
$t[u/\cst{p}]$ replaces $\cst{p}$ by a term with the same semantics
according to $\III'$, the interpretation $\III'$ is proper.

We are now ready to tackle the three conditions we need to prove.
To prove (1), we start from $\III \models \withoutpred{N}{\cst{p}}$ and show
$\III' \models \withoutpred{N}{\cst{p}}$. More precisely, we must show that
$\sem{\withoutpred{N}{\cst{p}}}{\III'}{\xi}$ is true for every valuation $\xi$.
This is obvious because $\III$ and $\III'$ only differ on $\cst{p}$,
which does not occur in $\withoutpred{N}{\cst{p}}$.

To prove (2), we show that $\III' \models \withpredpos{N}{\cst{p}}$ holds by
construction of $\III'$. More precisely, we must show that
$\sem{\withpredpos{N}{\cst{p}}}{\III'}{\xi}$ is true for every valuation $\xi$. Let
$C = \pospredlit{\cst{p}\tyargs{\tuple{\upsilon}}\>\tuple{t}} \cor C'$ be a clause in
$\withpredpos{N}{\cst{p}}$.
By definition,
\[\II'(\cst{p}, \sem{\tuple{\upsilon}}{\IIIty}{\xi})
= \sem{\lambda\tuple{x}.\>
\cdots \ior \underbrace{(\iexists\tuple{y}.\> x_1 \ieq t_1 \iand \cdots \iand x_n \ieq t_n \iand \inot\> \iform{C'})}_{\smash{\textstyle\psi}}
\ior\allowbreak \cdots}{\III}{\xi}\]
If $\sem{C'}{\III'}{\xi}$ is true, then clearly $\sem{C}{\III'}{\xi}$ is true.
Otherwise, the literal $\pospredlit{\cst{p}\tyargs{\tuple{\upsilon}}\>\tuple{t}}$
is true, because the disjunct $\psi$ is true. This can be seen by taking the
values of $\tuple{y}$ under $\xi$ as the existential witnesses. The arguments
$\tuple{t}$ in $\pospredlit{\cst{p}\tyargs{\tuple{\upsilon}}\>\tuple{t}}$ passed
for $\tuple{x}$ match those expected by the equalities $x_i \ieq t_i$,
and since $\sem{C'}{\III'}{\xi}$ is false, $\sem{\inot\> \iform{C'}}{\III'}{\xi}$
is true.

For the proof of (3),
the argument is similar to Khasidashvili and Korovin's
\cite[Theorem~1]{kk-2016-pe-fol}. It relies on the presence of the
flat resolvents
$\withpredpos{N}{\cst{p}} \flatres_{\!\cst{p}} \withpredneg{N}{\cst{p}}$
in the result set.
Let $D = \lnot\>\cst{p}\tyargs{\tuple{\upsilon}}\>\tuple{u} \lor D'$ be a clause
in $\withpredneg{N}{\cst{p}}$.
If $\sem{D'}{\III}{\xi}$ is true, then clearly $\sem{D'}{\III'}{\xi}$ is true.
Otherwise, $\sem{D'}{\III'}{\xi}$ is false, and
$\sem{\lnot\>\cst{p}\tyargs{\tuple{\upsilon}}\>\tuple{u}}{\III'}{\xi}$
is either true or false. In the true case, $\sem{D}{\III'}{\xi}$ is true, as desired.
As for the remaining case, it is impossible for the following reason.
If $\sem{\lnot\>\cst{p}\tyargs{\tuple{\upsilon}}\>\tuple{u}}{\III'}{\xi}$
were false, this would mean
$\sem{\cst{p}\tyargs{\tuple{\upsilon}}\>\tuple{u}}{\III'}{\xi}$
is true. By definition of $\III'$, this would then mean
$\semleft
\cdots \ior (\iexists\tuple{y}.\> u_1 \ieq t_1 \iand \cdots \iand u_n \ieq t_n \iand \inot\> \iform{C'})
\ior\allowbreak \cdots \semright{\III}{\xi}$
is true.

Suppose the displayed disjunct is one of those that makes the whole big
disjunction true. This means that there
exists a clause $C = \cst{p}\tyargs{\tuple{\tau}}\>\tuple{t} \lor C'$
in $\withpredpos{N}{\cst{p}}$ such that
$\sem{\tuple{t}}{\III'}{\xi} = \sem{\tuple{u}}{\III'}{\xi}$,
and $\sem{C'}{\III'}{\xi}$ is false,
where $\tuple{\tau}$ is unifiable with $\tuple{\upsilon}$. Let $\sigma$
be the most general unifier of $\tuple{\tau} \UNIF \tuple{\upsilon}$.
If $C$ exists, the flat resolvent $(t_1 \noteq u_1 \cor \cdots
\cor t_n \noteq u_n \cor C' \cor D')\sigma$
of $C$ and $D$ must be false in $\III'$ under $\xi$, and since it does not contain $\cst{p}$,
it would be false in $\III$ under $\xi$ as well, contradicting the hypothesis
that $\III$ satisfies it.
\end{proof}

As Khasidashvili and Korovin observed, eliminating all singular predicates
indiscriminately can dramatically increase the number of clauses in the problem.
To prevent this explosion, Vukmirovi\'c et al.\ proposed the following criterion.
Let $K_\mathrm{tol} \in \mathbb{N}$ be a tolerance
parameter. The application of SPE from
$N$ to $N'$ is allowed if
$\lambda(N') < \lambda(N) + K_\mathrm{tol}$ or
$\mu(N') < \mu(N)$ or
$|N'| < |N| + K_\mathrm{tol}$,
where
$\lambda(N)$ is the number of literals in~$N$ and
$\mu(N)$ is the sum, for all clauses $C \in
N$, of the square of the number of unique variables in $C$.

\oursubsection{Defined Predicate Elimination}

The next technique we generalize from first-order logic to higher-order logic is
called defined predicate elimination. It generalizes, in turn, the propositional
technique of elimination by substitution \cite{eb-2005-satpreprocess}.

Given a clause set $N$,
the basic idea is that the set $\withpred{N}{\cst{p}}$ of clauses containing
$\cst{p}$ is partitioned between a definition set (or ``gate'') $G$ and the
remaining clauses $R$. The definition set fully characterizes $\cst{p}$ for all
input in a unique way and can be seen as constituting a definition of the form
$\cst{p}\>\tuple{x} \iequiv \varphi$, where the variables
$\tuple{x}$ are distinct, $\cst{p}$ does not occur in
$\varphi$, and the variables in $\varphi$ are all among $\tuple{x}$. Because
of clausification, $G$ will usually consist of multiple clauses that together
are equivalent to $\cst{p}\>\tuple{x} \iequiv \varphi$ for some
$\varphi$.

We define definition sets largely as in monomorphic first-order logic, but with
additional requirements on the type arguments and type variables.

\begin{defi}
\label{def:definition-set}
\looseness=-1
Let $G$ be a clause set and $\cst{p}$ be a predicate symbol.
The set $G$ is a \emph{definition set} for $\cst{p}$ if
\begin{enumerate}
\item \label{itm:definition-set-singular}
  $\cst{p}$ is singular for $G$;
\item \label{itm:definition-set-form}
  $G$ consists of clauses of the form
  $\arbpredlit{\cst{p}\tyargs{\tuple{\alpha}}\>\tuple{x}} \cor C'$
  up to variable renaming,
  where $\tuple{\alpha}$ are distinct type variables
  and $\tuple{x}$ are distinct term variables;
\item \label{itm:definition-set-tvars}
  the type variables in $C'$ are all among $\tuple{\alpha}$;
\item \label{itm:definition-set-vars}
  the term variables in $C'$ are all among $\tuple{x}$;
\item \label{itm:definition-set-tauto}
  all clauses in $\withpredpos{G}{\cst{p}}
  \flatres_{\!\cst{p}} \withpredneg{G}{\cst{p}}$ are tautologies; and
\item \label{itm:definition-set-unsat}
  $E(\tuple\iota, \tuple{\cst{c}})$ is unsatisfiable,
  where the \emph{environment} $E(\tuple{\alpha}, \tuple{x})$ consists of all subclauses $C'$ of
  any $\arbpredlit{\cst{p}\tyargs{\tuple{\alpha}}\>\tuple{x}} \mathbin{\cor} C'
  \in G$,
  $\tuple{\iota}$ is a tuple of distinct nullary type constructors
  substituted in for $\tuple{\alpha}$,
  and $\tuple{\cst{c}}$ is a tuple of distinct fresh symbols
  substituted in for $\tuple{x}$.
\end{enumerate}
\end{defi}

Intuitively, conditions
\ref{itm:definition-set-singular}~and~\ref{itm:definition-set-form} check that
the definition set looks like the clausification of a definition; conditions
\ref{itm:definition-set-tvars}, \ref{itm:definition-set-vars}, and
\ref{itm:definition-set-tauto} check that the definition is not overconstrained
(e.g., $\cst{p}\>\cst{a}$ is not required to be both true and false); and
condition~\ref{itm:definition-set-unsat} checks that it is not underconstrained
(e.g., $\cst{p}\>\cst{a}$ is not unspecified).

\begin{defi}
\label{def:assoc-def}
Given a definition set $G$ for $\cst{p}$, its \emph{associated definition} is
the formula $\cst{p}\tyargs{\tuple{\alpha}}\>\tuple{x} \iequiv \varphi$, up to
variable renaming, where $\varphi$ is the disjunction $\varphi_1 \lor \cdots \lor
\varphi_n$ of all formulas $\varphi_j$ of the form $\inot\> \iform{C'}$ such that
$\pospredlit{\cst{p}\tyargs{\tuple{\alpha}}\>\tuple{x}}
\cor C'$ is contained in $G$ up to variable renaming.
\end{defi}

Note that our notion of definition set ensures that in Definition~\ref{def:assoc-def}
the type variables $\tuple{\alpha}$ are distinct, the term variables $\tuple{x}$
are distinct, $\cst{p}$ does not occur in $\varphi$, the type variables in
$\varphi$ are all among $\tuple{\alpha}$, and the variables in $\varphi$ are all
among $\tuple{x}$.

\begin{exa}
\label{ex:tautologies}
For the formula $\cst{p}\>x\> y \iequiv
\cst{q}\>x \cor \cst{r}\>y$, clausification would produce the clause set $\{
\negpredlit{\cst{p}\>x\> y} \cor \pospredlit{\cst{q}\>x} \cor \pospredlit{\cst{r}\>y}{,}\;
\pospredlit{\cst{p}\>x\> y} \cor \negpredlit{\cst{q}\>x}{,}\;
\pospredlit{\cst{p}\>x\> y} \cor \negpredlit{\cst{r}\>y}\}$,
which qualifies
as a definition set for $\cst{p}$. The associated definition is
$\cst{p}\>x\> y \iequiv \inot\,\inot\>\cst{q}\>x \ior \inot\,\inot\>\cst{r}\>y$.
\end{exa}

\begin{lem}
\label{lem:definition-set-equivalent-associated-definition}
Let $G$ be a definition set for $\cst{p}$. Then $G$ is equivalent to the
definition $\cst{p}\tyargs{\tuple{\alpha}}\>\tuple{x} \iequiv \varphi$
associated with $G$.
\end{lem}

\begin{proof}
We will show that under every valuation $\xi$, any model of $G$ is a model of
$\cst{p}\tyargs{\tuple{\alpha}}\>\tuple{x} \iequiv \varphi$ and vice versa.

Let $\III \models G$. We will show that
$\sem{\cst{p}\tyargs{\tuple{\alpha}}\>\tuple{x}}{\III}{\xi} =
\sem{\varphi}{\III}{\xi}$. If
$\sem{\cst{p}\tyargs{\tuple{\alpha}}\>\tuple{x}}{\III}{\xi}$ is false, then for
each clause $\pospredlit{\cst{p}\tyargs{\tuple{\alpha}}\>\tuple{x}} \cor C' \in
G$, we have that $\sem{C'}{\III}{\xi}$ must be true. This in turn makes the
right-hand side $\varphi$ false.
Otherwise, $\sem{\cst{p}\tyargs{\tuple{\alpha}}\>\tuple{x}}{\III}{\xi}$ is true.
Then for each clause
$\negpredlit{\cst{p}\tyargs{\tuple{\alpha}}\>\tuple{x}} \cor C' \in G$, we have
that $\sem{C'}{\III}{\xi}$ must be true. By
condition~\ref{itm:definition-set-unsat}, there must exist a clause
$\pospredlit{\cst{p}\tyargs{\tuple{\alpha}}\>\tuple{x}} \cor D' \in G$ such that
$\sem{D'}{\III}{\xi}$ is false. This means that $\sem{\inot\> D'}{\III}{\xi}$ is
true and hence the entire disjunction $\varphi$ is true, as desired.

Now let $\III \models \cst{p}\tyargs{\tuple{\alpha}}\>\tuple{x} \iequiv
\varphi$. We need to show that $\III \models G$. We will first prove the case of
clauses in $G$ where the $\cst{p}$-literal is positive; then we will consider
the negative case.
Let $C = \pospredlit{\cst{p}\tyargs{\tuple{\alpha}}\>\tuple{x}} \cor C' \in G$.
If $\pospredlit{\cst{p}\tyargs{\tuple{\alpha}}\>\tuple{x}}$ is true, then $C$
is true, as desired. Otherwise, from
$\III \models \cst{p}\tyargs{\tuple{\alpha}}\>\tuple{x} \iequiv
\varphi$ we have that $\varphi$ is false. This means that all of its disjuncts
are false and hence that $C'$ is true, meaning that $C$ is true.
For the remaining case, let $C =
\negpredlit{\cst{p}\tyargs{\tuple{\alpha}}\>\tuple{x}} \cor C' \in G$. If
$\pospredlit{\cst{p}\tyargs{\tuple{\alpha}}\>\tuple{x}}$ is false, then $C$ is
true, as desired. Otherwise, from
$\III \models \cst{p}\tyargs{\tuple{\alpha}}\>\tuple{x} \iequiv
\varphi$ we have that $\varphi$ is true. This means that there exists a
clause $D = \pospredlit{\cst{p}\tyargs{\tuple{\alpha}}\>\tuple{x}} \cor D' \in G$
such that $\sem{D'}{\III}{\xi}$ is false. By
condition~\ref{itm:definition-set-tauto}, the resolvent $C' \cor D'$ of $C$ and
$D$ must be a tautology. Hence $C'$ is true and thus $C$ is true, as desired.
\end{proof}

Once a definition is identified, it is expanded in the remaining clauses $R$.
For $\cst{p}$-literals in $R$, this is achieved as in first-order logic using
flat resolution. For deeper occurrences of $\cst{p}\tyargs{\tuple{\tau}}$ in
$R$, which may arise in higher-order logic, this is achieved by replacing them
by the $\lambda$-abstraction
$\lambda\tuple{x}.\> \varphi\{\tuple{\alpha}\mapsto\tuple{\tau}\}$. An
alternative would be to replace \emph{all} occurrences of
$\cst{p}\tyargs{\tuple{\tau}}$ and not only deep occurrences, but this would
leave more work for the clausifier.

\begin{defi}
\looseness=-1
Let $N$ be a clause set and $\cst{p}$ be a predicate symbol.
Let $G \subseteq N$ be a definition set
for $\cst{p}$ with associated definition
$\cst{p}\tyargs{\tuple{\alpha}}\>\tuple{x} \iequiv \varphi$.
Let $R = \withpred{N}{\cst{p}} \setminus G$.
\emph{Defined predicate elimination} (DPE) of $\cst{p}$ in $N$ replaces $N$ by
$\withoutpred{N}{\cst{p}} \cup
(G \flatres_{\!\cst{p}} R)
[\cst{p}\tyargs{\tuple{\tau}} \allowbreak\mapsto \lambda\tuple{x}.\> \varphi\{\tuple{\alpha}\mapsto\tuple{\tau}\}]$.
\end{defi}

The key result is that DPE preserves satisfiability and unsatisfiability. The
proof builds on three lemmas.

\begin{lem}
\label{lem:flat-res-set-step-satisfiability}
Let $G$ be a definition set for $\cst{p}$ and $R$ be an arbitrary clause
set. If $(G, R) \flatresiter (G, R')$, then $G \cup R$ and $G \cup R'$ are
equivalent.
\end{lem}

\begin{proof}
The proof is essentially as in the first-order case
\cite[Lemma~4.11]{vukmirovic-et-al-2023-sat}.
\end{proof}

\begin{lem}
\label{lem:flat-res-set-last-replace-satisfiability}
Let $G$ be a definition set for $\cst{p}$ with associated definition
$\cst{p}\tyargs{\tuple{\alpha}}\>\tuple{x} \iequiv \varphi$, and let $R$ be a
clause set.
Then $G \cup R$ and $G \cup R[\cst{p}\tyargs{\tuple{\tau}} \allowbreak\mapsto\allowbreak
\lambda\tuple{x}.\> \varphi\{\tuple{\alpha}\mapsto\tuple{\tau}\}]$ are
equivalent.
\end{lem}

\begin{proof}
By Lemma~\ref{lem:definition-set-equivalent-associated-definition},
$G$ entails the characterization $\cst{p}\tyargs{\tuple{\alpha}}\>\tuple{x}
\iequiv \varphi$.
Hence, by functional extensionality,
$G$ entails $\cst{p}\tyargs{\tuple{\alpha}} \ieq
\lambda\tuple{x}.\> \varphi$.
Thus, in any model of $G$,
$\cst{p}\tyargs{\tuple{\alpha}}$ has the same interpretation as
$\lambda\tuple{x}.\> \varphi$.
In particular, this applies to their instances:
$\cst{p}\tyargs{\tuple{\tau}}$ has the same interpretation as
$\lambda\tuple{x}.\> \varphi\{\tuple{\alpha}\mapsto\tuple{\tau}\}$.
\end{proof}

\begin{lem}
\label{lem:flat-res-set-last-step-satisfiability}
Let $G$ be a definition set for $\cst{p}$ and $R$ be a clause set with no
occurrences of $\cst{p}$. Then $G \cup R$ is satisfiable
if and only if $R$ is satisfiable.
\end{lem}

\begin{proof}
The proof is essentially as in the first-order case
\cite[Lemma~4.12]{vukmirovic-et-al-2023-sat}.
\end{proof}

\begin{thm}
\label{thm:dpe-equisat}
The result of applying DPE to a finite clause set $N$ is
satisfiable if and only if $N$ is satisfiable.
\end{thm}

\begin{proof}
Let $\cst{p}$ be a predicate symbol and $G \subseteq N$ be the
definition set used by DPE. Let $R = \withpred{N}{\cst{p}} \setminus G$.
The core of DPE is the computation of $G \flatres_{\!\cst{p}} R$, via a
derivation $(G,R) \flatresiter^n (G,R') \flatresiter R'$.
Applying Lemma~\ref{lem:flat-res-set-step-satisfiability} $n$ times,
we get that $G \cup R$ is equivalent to $G \cup R'$.
Moreover, by
Lemma \ref{lem:flat-res-set-last-replace-satisfiability},
$G \cup R'$ is equivalent to $G \cup R'[\cst{p}\tyargs{\tuple{\tau}} \mapsto
\lambda\tuple{x}.\> \varphi\{\tuple{\alpha}\mapsto\tuple{\tau}\}]$.
Finally, by
Lemma \ref{lem:flat-res-set-last-step-satisfiability},
$G \cup R'[\cst{p}\tyargs{\tuple{\tau}} \mapsto \lambda\tuple{x}.\>
\varphi\{\tuple{\alpha}\mapsto\tuple{\tau}\}]$ is equivalent to
$R'[\cst{p}\tyargs{\tuple{\tau}} \mapsto \lambda\tuple{x}.\>
\varphi\{\tuple{\alpha}\mapsto\tuple{\tau}\}]$.
\end{proof}

\oursubsection{Portfolio Predicate Elimination}

A reasonable strategy for applying predicate elimination is to use a portfolio
of DPE and SPE, first trying to apply DPE and, if this fails, trying SPE as a
fallback.

\begin{defi}
Let $N$ be a clause set and $\cst{p}$ be a predicate symbol. If there exists a
definition set $G \subseteq N$ for $\cst{p}$, \emph{portfolio predicate
elimination} (PPE) on $\cst{p}$ applies DPE on $\cst{p}$.
Otherwise, if $\cst{p}$ is singular in $N$, PPE applies SPE on $\cst{p}$.
In all other cases, PPE is not applicable.
\end{defi}

Like SPE and DPE (Theorems \ref{thm:spe-equisat}~and~\ref{thm:dpe-equisat}), PPE
can be used as a preprocessor without affecting satisfiability. As for
inprocessing, Vukmirovi\'c et al.\ \cite{vukmirovic-et-al-2023-sat} explained
that under a reasonable condition, the first-order version of PPE can be used at
any point during proof search in a superposition prover without compromising
refutational completeness. Inspection of the proofs reveals that the same
applies to higher-order PPE and $\lambda$-superposition.

\section{Blocked Clause Elimination}
\label{sec:blocked-clause-elimination}

In propositional logic, a powerful technique for simplifying a clause set is to
identify and remove so-called blocked clauses. These are clauses whose
resolvents with other clauses in the set are all tautologies. Removing such
clauses preserves unsatisfiability.
Blocked clause elimination has been extended to first-order logic with equality
by Kiesl et al.\ \cite{ksstb-2017-blockedfol}. They call their key notion
``equality-blocked clauses,'' but since we consider only a logic with equality,
we simply call these clauses ``blocked.''

Blocked clause elimination has been shown by Vukmirovi\'c et al.\
\cite[Section~5]{vukmirovic-et-al-2023-sat} to be incompatible with the
saturation loop of a superposition prover. Nevertheless, the technique can still
be used effectively as a preprocessor, or even as an inprocessing technique
within the prover's saturation loop at the cost of potential divergence on some
unsatisfiable problems.

Our extension to polymorphic higher-order logic is based on a slightly weaker
definition of blocked clause than Kiesl et al. We were unsuccessful at showing
that a generalization of blocked clause elimination based on their concept
preserves unsatisfiability with respect to general interpretations. The notion
we propose allows the generalization.

\begin{defi}
\label{def:flat-l-resolvent}
Let $C = L \cor C'$ and $D = L' \cor D'$ be clauses such that
\begin{enumerate}
\item\label{itm:flatres-atm-l}
  the atom of $L$ is $\cst{p}\tyargs{\tuple{\tau}}\>\tupleii{s}{n}$;
\item\label{itm:flatres-atm-l'}
  the atom of $L'$ is $\cst{p}\tyargs{\tuple{\upsilon}}\>\tupleii{t}{n}$;
\item\label{itm:flatres-opposite-pol}
  the literal $L'$ is of opposite polarity to $L$;
\item\label{itm:flatres-dist-vars}
  $C$ and $D$ have no (type or term) variables in common; and
\item\label{itm:mgu-tau-upsilon}
  $\sigma$ is the most general unifier of $\tuple{\tau} \UNIF \tuple{\upsilon}$.
\end{enumerate}
The clause
$\bigl(\bigl(\bigvee_{\smash{j=}1}^n
 \allowbreak\neqlit{s_j}{t_{j}}\bigr) \cor C' \cor D'\bigr)\sigma$
is a \emph{binary flat $L$-resolvent} of $C$~and~$D$.
\end{defi}

We already see a first key difference with Keisl et al.: They consider $n$-ary
flat resolvents, whereas we need to consider only binary resolvents, for reasons
illustrated below (Example~\ref{ex:two-ps-nonexample}). Another, more superficial
difference is that our definition is generalized to polymorphic higher-order
logic.

\begin{defi}
\label{def:blocked}
Let $L = (\lnot)\,\cst{p}\tyargs{\tuple{\tau}}\>\tuple{s}$ be a predicate
literal, $C = L \cor C'$ be a clause, and $N$ be a clause set. Let $N'$ consist
of all clauses from $N \setminus \{C\}$ with their type and term variables
renamed so that $N'$ shares no variables with $C$. The clause $C$ is
\emph{blocked} by $L$ in the set $N$ if the following conditions are met:
\begin{enumerate}
\item\label{itm:blocked-poly-safe}
  $C$ is polymorphism-safe for $L$;
\item\label{itm:blocked-no-deep-p}
  $N$ contains no deep occurrences of $\cst{p}$;
\item\label{itm:blocked-c'-no-same-pol}
  $C'$ contains no $\cst{p}$-literals with the same polarity as $L$; and
\item\label{itm:blocked-tauto}
  all binary flat $L$-resolvents between $C$ and clauses in $N'$ are tautologies.
\end{enumerate}
\end{defi}

We now see another key difference with Keisl et al.: They have no restriction
corresponding to condition~\ref{itm:blocked-c'-no-same-pol} of
Definition~\ref{def:blocked}. In this respect, our notion is less powerful
than theirs. (They also have no restriction corresponding to conditions
\ref{itm:blocked-poly-safe} and \ref{itm:blocked-no-deep-p}, but these
conditions are trivially satisfied in a monomorphic first-order setting.)
As a result, their notion and our notion of blocked clause are incomparable in
strength.

\begin{exa}
This example is based on Keisl et al.~\cite[Example~1]{ksstb-2017-blockedfol}.
Let $C = \lnot\>\cst{p} \cor\nobreak \cst{q}$, and take $N = \{C{,}\; \cst{p}
\cor\nobreak \lnot\>\cst{q}{,}\; \lnot\>\cst{q} \cor \cst{r}\}$ as the clause
set. The clause $C$ is blocked by $\lnot\>\cst{p}$ in $N$ according to
Definition~\ref{def:blocked} because the only resolvent of $C$ on
$\lnot\>\cst{p}$ is the tautology $\cst{q} \cor \lnot\>\cst{q}$ resulting from
resolution against $\cst{p} \cor \lnot\>\cst{q}$. The clause $C$ is also blocked
according to the definition in Kiesl et al.
\end{exa}

\begin{exa}
\label{ex:two-ps-nonexample}
The next example is also based on Keisl et
al.~\cite[Example~4]{ksstb-2017-blockedfol}:\
$C = \cst{p}\> x\> y \cor \cst{p}\> y\> x$,
$D = \lnot\> \cst{p}\> x\> y \cor \lnot\> \cst{p}\> y\> x$,
and $N = \{C, D\}$.
The set $N$ is unsatisfiable, because $C$ entails $\cst{p}\>x\>x$ and $D$
entails $\lnot\>\cst{p}\>x\>x$. On the other hand, $D$ alone is satisfiable.
Hence, removing $C$ from $N$ does \emph{not} preserve unsatisfiability, and
therefore $C$ should \emph{not} be considered blocked. With
Definition~\ref{def:blocked}, $C$ correctly cannot be blocked on a
$\cst{p}$-literal by condition~\ref{itm:blocked-c'-no-same-pol}, because of the
presence of another $\cst{p}$-literal in the clause. With Kiesl et al., this
condition is missing, but since they consider all $n$-ary
resolvents, the nontautological resolvent $\cst{p}\> x\> x$ is
computed. In both cases, $C$ is correctly considered not blocked.
\end{exa}

\begin{exa}
Let $C = \cst{p}\> \cst{a} \cor \cst{p}\> \cst{b} \cor \lnot\> \cst{q}$,
$D = \lnot\> \cst{p}\> x \cor \cst{q}$, and $N = \{C, D\}$. With
Definition~\ref{def:blocked}, condition~\ref{itm:blocked-c'-no-same-pol} prevents
$C$ from being considered blocked. In contrast, the clause is considered blocked
by Kiesl et al.
\end{exa}

We will now show that removing a blocked clause from a clause set preserves the
set's unsatisfiability. Our strategy is loosely inspired by Kiesl et
al.~\cite[Section~4]{ksstb-2017-blockedfol}.

\begin{defi}
\label{def:flipping}
Let $\III = (\IIIty, \II, \LL)$ be an interpretation
and $L \lor C'$ be a clause that is
polymorphism-safe for $L$ and where $L =
(\lnot)\,\cst{p}\tyargs{\tuple{\tau}}\>\tupleii{s}{n}$.
Let $\tuple{y}$ be the tuple of all free variables in $L \lor C'$ and
$\tupleii{x}{n}$ be a tuple of fresh variables.
The interpretation $\III^\star = (\IIIty, \II^\star, \LL^\star)$ obtained by
\emph{flipping} the truth value of $L$
in $L \lor C'$ is defined as follows by mutual recursion.
We let $\II^\star$ be the function defined as follows:
\[\II^\star(\cst{f}, \tuple{\DD}) = \begin{cases}
\sem{\lambda\tupleii{x}{n}.\; \varphi}{\III}{\xi}
  & \text{if $\cst{f} = \cst{p}$ and $\sem{\tuple{\tau}}{\IIIty}{\xi} = \tuple{\DD}$
    for some type valuation $\xi$} \\
\II(\cst{f}, \tuple{\DD})
  & \text{otherwise}
\end{cases}\]
where $\varphi$ is defined by
\[
\varphi =
\begin{cases}
  \cst{p}\tyargs{\tuple{\tau}}\>\tuple{x} \mathrel{\ior}
    (\iexists\tuple{y}.\> x_1 \ieq s_1 \iand \cdots \iand x_n \ieq s_n \iand \inot\>\iform{C'})
    & \text{if $L$ is positive} \\
  \cst{p}\tyargs{\tuple{\tau}}\>\tuple{x} \mathrel{\iand}
    (\iforall\tuple{y}.\> x_1 \ieq s_1 \iand \cdots \iand x_n \ieq s_n \iimplies \iform{C'})
    & \text{if $L$ is negative}
\end{cases}
\]
Moreover, for all valuations
$\xi$ and $\lambda$-abstractions~$\lambda x \mathbin: \upsilon.\>t$, we let
$\LL^\star(\xi,\lambda x.\>t)$ be the function that maps each $v \in
\sem{\upsilon}{\IIIty}{\xi}$ to $\sem{t}{\III^\star}{\xi[x\mapsto v]}$.
\end{defi}

This definition introduces a well-formed interpretation. Because $L \lor C'$ is
polymorphism-safe for $L$,
the right-hand side
$\sem{\lambda\tupleii{x}{n}.\; \varphi}{\III}{\xi}$ of
$\II^\star(\cst{f}, \tuple{\DD})$ is uniquely defined. Moreover, the
comprehension principle guarantees that the corresponding predicate exists in
the interpretation of $\cst{p}$'s type.
Similarly, the function that provides the interpretation for a
$\lambda$-abstraction $\lambda x.\>t$ exists in the domain associated with the
$\lambda$-abstraction's type. This is because the semantics in $\III^\star$ of any
occurrences of $\cst{p}$ in $t$ corresponds to the semantics in $\III$ of
$\lambda\tupleii{x}{n}.\; \varphi$, and $\III$ is a well-formed
interpretation.

The intuition behind $\III^\star$ is that whenever the clause $L \lor C'$ is
blocked and $\III \models N \setminus \{L \lor C'\}$, we have $\III^\star
\models N$. Since adding the blocked clause preserves satisfiability,
removing it preserves unsatisfiability.

\begin{exa}
\label{ex:flipping-twice}
We will try to justify the definition of $\varphi$ above with an example.
Consider the clause set $\{\cst{p}\>\cst{a}{,}\;\,
\lnot\> \cst{p}\>z \cor z \eq \cst{a} \cor z \eq \cst{b}\}$
and an interpretation $\III$ that maps $\cst{p}$ to the uniformly false
predicate. Clearly, $\III$ is not a model of the first clause, $\cst{p}\>\cst{a}$.
The interpretation $\III^\star$ obtained by
flipping the truth value of the literal $\cst{p}\>\cst{a}$
in the first clause
interprets $\cst{p}$ in the same way as $\III$ interprets
\[\lambda x.\> \cst{p}\>x \mathrel{\ior} (x \ieq \cst{a} \mathrel{\iand} \inot\> \ifalse)\]
In other words, $\III^\star$ makes $\cst{p}$ true for arguments interpreted as equal
to $\cst{a}$ and false otherwise. Intuitively, the interpretation $\III$ is
flipped to satisfy the clause $\cst{p}\>\cst{a}$.

Next, assume instead that $\III$ interprets $\cst{p}$ as the uniformly true
predicate, and consider the interpretation $\III^\star$
obtained by
flipping the truth value of $\lnot\> \cst{p}\>z$
in the second clause,
$\lnot\> \cst{p}\>z \cor z \eq \cst{a} \cor z \eq \cst{b}$.
Then $\III^\star$ interprets $\cst{p}$ in the same way as $\III$ interprets
\[\lambda x.\> \cst{p}\>x \mathrel{\iand}
  (\iforall z.\> x \ieq z \mathrel{\iimplies} z \ieq \cst{a}
  \mathrel{\ior} z \ieq \cst{b})\]
This means that $\III^\star$ makes $\cst{p}$ true for arguments interpreted as
equal to $\cst{a}$ or $\cst{b}$ and false otherwise. Intuitively, the interpretation $\III$
is flipped to satisfy the clause $\lnot\> \cst{p}\>z \cor z \eq \cst{a}
\cor z \eq \cst{b}$; whenever $z \noteq \cst{a}$ and $z \noteq \cst{b}$, we
have $\lnot\> \cst{p}\>z$.
\end{exa}

\begin{lem}
\label{lem:bce-clause-equisat-flipped-interp}
Let $N$ be a clause set and $C$ be a clause
contained in $N$ such that $C$ has no variables in common with $N \setminus
\{C\}$. Assume $C$ is blocked by $L$ in $N$.
Let $\III$ be an interpretation and
$\III^\star$ be the interpretation obtained by flipping the truth value of $L$.
For every $D \in N \setminus \{C\}$,
if $\III \models D$, then $\III^\star \models D$.
\end{lem}

\begin{proof}
Let $C = L \lor C'$ where $L =
(\lnot)\,\cst{p}\tyargs{\tuple{\tau}}\>\tupleii{s}{n}$
(as per condition~\ref{itm:flatres-atm-l} of Definition~\ref{def:flat-l-resolvent}).
Assuming that $\III \models D$,
we will show that $\sem{D}{\III^\star}{\xi}$ is true for every valuation $\xi$.

Let $L'$ be a literal of $D$ of opposite polarity to $L$, whose atom is
$\cst{p}\tyargs{\tuple{\upsilon}}\>\tuple{t}$, and such that, for each $j$,
$\sem{L'}{\III}{\xi}$ is true and $\sem{L'}{\III^\star}{\xi}$ is false.
Intuitively, $L'$ is a literal whose truth value goes from true in $\III$ to
false in $\III^\star$. We distinguish two cases: The case where such a literal $L'$
exists and the case where it does not.

\medskip

\noindent
\textsc{Case where $L'$ does not exist:}\enskip
If $D$ contains literals $L'$ of opposite polarity to $L$, then either
$\sem{L'}{\III}{\xi}$ is false and then the truth value of
$\sem{L'}{\III^\star}{\xi}$ is irrelevant or $\sem{L'}{\III}{\xi}$ is true and
then $\sem{L'}{\III^\star}{\xi}$ is true. As for the occurrences of $\cst{p}$ in $D$ that
have the same polarity as $L$, if these were true in $\III$, they are by
construction also true in $\III^\star$.
Finally, $D$ contains no deep
occurrences of $\cst{p}$ (by condition \ref{itm:blocked-no-deep-p} of
Definition~\ref{def:blocked}). 
Since all the literals either keep their
truth value or go from false to true when taking the step from $\III$ to $\III^\star$,
we have that $\sem{D}{\III^\star}{\xi}$ is true.

\medskip

\noindent\textsc{Case where $L'$ exists:}\enskip
We will show that $\sem{D}{\III^\star}{\xi} = \sem{D}{\III}{\xi}$
even if $L'$ has gone from true in $\III$ to false in~$\III^\star$ under
$\xi$. Since $\III \models D$, $\sem{D}{\III^\star}{\xi}$ will then be true.

Since $\III^\star$ flips the truth value of $L'$, by construction of $\III^\star$,
this flipping must be triggered by $C$. Hence, there must exist a
valuation $\xi'$ such that
$\sem{\tau_i}{\III}{\xi'} = \sem{\upsilon_{i}}{\III}{\xi}$ and
$\sem{s_j}{\III}{\xi'} = \sem{t_{j}}{\III}{\xi}$
for every $i, j$.
Let $\xi''$ be the valuation that coincides with $\xi'$ on $C$'s free variables
and with $\xi$ on $D$'s free variables.
Let $D = L' \lor D'$.

Since $C$ is blocked by $L$ in $N$,
all binary flat $L$-resolvents of $C$ with clauses from $N$ are tautologies
(by condition~\ref{itm:blocked-tauto} of Definition~\ref{def:blocked}).
In particular, consider the binary flat $L$-resolvent of $C$ and $D$ of the form
$\bigl(\bigl(\bigvee_{\smash{j=}1}^n
 \allowbreak\neqlit{s_j}{t_j}\bigr) \cor C' \cor D'\bigr)\sigma$,
where $\sigma$ is the most general unifier of
$\tuple{\tau} \UNIF \tuple{\upsilon}$.
This binary flat $L$-resolvent, which must exist by the five conditions of
Definition~\ref{def:flat-l-resolvent},
is a tautology,
and it must be satisfied by $\III^\star$ under $\xi$.

By definition of $\xi''$, $\sem{\neqlit{s_j}{t_{j}}}{\III}{\xi''}$
must be false, and since the terms $\tuple{s}, \tuple{t}$ do not contain $\cst{p}$
(by condition \ref{itm:blocked-no-deep-p} of
Definition~\ref{def:blocked}),
$\sem{\neqlit{s_j}{t_{ij}}}{\III^\star}{\xi''}$ must be false as well.
Moreover, by construction of $\III^\star$,
the only way for the interpretation
of $\cst{p}\tyargs{\tuple{\tau}}\>\tupleii{s}{n}$ in
$\III^\star$ under $\xi''$ to differ from that in $\III$ under $\xi''$ is if
$\sem{C'}{\III}{\xi''}$ is false, and since $C'$ contains only
occurrences of $\cst{p}$ of opposite polarity to $L$
(by conditions \ref{itm:blocked-no-deep-p}~and~\ref{itm:blocked-c'-no-same-pol}
of Definition~\ref{def:blocked}), $\sem{C'}{\III^\star}{\xi''}$ must still be
false after we flipped $L$ to make it true.

Finally, since both $\sem{\neqlit{s_j}{t_{j}}}{\III}{\xi''}$ and
$\sem{C'}{\III^\star}{\xi''}$ are false and the binary flat $L$-resolvent
of $C$ and $D$ is a tautology,
$D'\sigma$ must be true in $\III^\star$ under every valuation.
Since $\sigma$ is a most general unifier and $\xi''$
assigns the same semantics to $\tuple{\tau}$ and $\tuple{\upsilon}$,
effectively ``unifying'' them,
we also have that $\sem{D'}{\III^\star}{\xi''}$ is true
and hence $\sem{D'}{\III^\star}{\xi}$ is true.
Thus $\sem{D}{\III^\star}{\xi}$ is true, as desired.
\end{proof}

\begin{lem}
\label{lem:bce-clause-equisat}
Let $N$ be a clause set and $C$ be a clause contained
in $N$. If $C$ is blocked by a literal in $N$, then $N \setminus \{C\}$ is
satisfiable if and only if $N$ is satisfiable.
\end{lem}

\begin{proof}
The ``if'' direction is trivial.
For the other direction, let $N'$ consist of all clauses from $N \setminus
\{C\}$ with their type and term variables renamed so that they share no
variables with $C$. Let $\III$ be a model of $N'$. By
Lemma~\ref{lem:bce-clause-equisat-flipped-interp}, $\III^\star$ is a model of
each $D \in N'$.

We also need to show that $\III^\star$ is a model of $C$.
Specifically, we must show that $\sem{C}{\III^\star}{\xi}$ is true for any
valuation~$\xi$. Let $C = L \lor C'$.
If $\sem{C'}{\III^\star}{\xi}$ is true, we are done. Otherwise,
first suppose $L$ is positive. Then $\xi$ provides the necessary witnesses for the
$\iexists$ quantifier in the definition of $\varphi$ in
Definition~\ref{def:flipping}, making the interpretation of $\cst{p}$ by $\III^\star$
true in that case, as in the first part of Example~\ref{ex:flipping-twice}.
Hence $\sem{L}{\III^\star}{\xi}$ is true, and thus $\sem{C}{\III^\star}{\xi}$ is true.
Next, suppose $L$ is negative. Then $\xi$ provides a counterexample to the
$\iforall$ quantifier in the definition of $\varphi$, making the interpretation
of $\cst{p}$ by $\III^\star$ false in that case, as in the second part of
Example~\ref{ex:flipping-twice}. Hence $\sem{L}{\III^\star}{\xi}$ is true, and thus
$\sem{C}{\III^\star}{\xi}$ is true.

Since $\III^\star \models N' \cup \{C\}$, we have that $N$ is satisfiable.
\end{proof}

\begin{defi}
Given a finite clause set, \emph{blocked clause elimination} (BCE) repeatedly
removes blocked clauses until no such clauses remain.
\end{defi}

The procedure is confluent and hence yields a unique result. This is easy to see
because removing a blocked clause will only make more clauses blocked; it can
never ``unblock'' a clause.

\begin{thm}
The result of applying BCE to a clause set $N$ is
satisfiable if and only if $N$ is satisfiable.
\end{thm}

\begin{proof}
This follows by iteration of Lemma~\ref{lem:bce-clause-equisat}.
\end{proof}

\section{Quasipure Literal Elimination}
\label{sec:quasipure-literal-elimination}

\emph{Pure literal elimination} (PLE) is one of the simplest optimizations
implemented in SAT solvers. It is a special case of variable elimination
\cite{sp-04-niver,cs-00-zres}: If a given variable always occurs with the same
polarity in a problem, the solver can assign it that polarity without loss of
generality, making all the clauses that contain it tautologies. PLE consists of
recursively deleting all such clauses.
PLE's generalization to first-order logic considers literals
$\arbpredlit{\cst{p}(\tuple{s})}$, where the arguments $\tuple{s}$ are
ignored by the analysis; only the polarity is considered. The same idea carries
over to higher-order logic.

\begin{exa}
\label{ex:ple}
Consider the clause set $N = \{ \cst{p}\>x \cor \cst{q}\>\cst{a}\> x{,}\allowbreak\;
\cst{p}\>(\cst{f}\>x){,}\allowbreak\; \neglit{\cst{q}\>\cst{a}\> \cst{a}} \}$.
Since $\cst{p}$ occurs only positively in $N$, PLE removes the two first
clauses. At that point, $\cst{q}$ occurs only negatively in the remaining
singleton clause set and can be removed as well. The result is the empty set,
which is obviously satisfiable, indicating that $N$ is satisfiable.
As model of $N$, we can take an interpretation $\III$ that makes all
$\cst{p}$-literals true and all $\cst{q}$-literals false.
\end{exa}

\begin{exa}
\label{ex:sle-i}
Consider the clause set $N' = \{
\cst{p}\>\cst{a}{,}\allowbreak\;
\neglit{\cst{p}\>x} \lor \cst{p}\>(\cst{f}\>x)
\}$. PLE does not apply because $\cst{p}$ occurs with both
polarities. Yet we notice that $\cst{p}$ occurs positively in
both clauses, and hence that the same reasoning as in Example~\ref{ex:ple}
applies: We can satisfy both clauses by making $\cst{p}$-literals true.
Using $\III$ from Example~\ref{ex:ple}, we have $\III \models N'$.
\end{exa}

\begin{exa}
\label{ex:sle-ii}
Consider the clause set $N'' = \{
\cst{p}\>\cst{a}{,}\allowbreak\;
\cst{q}\>x \lor \cst{p}\>(\cst{f}\>x){,}\allowbreak\;
\neglit{\cst{q}\>(\cst{f}\>\cst{a})}{,}\allowbreak\;
\neglit{\cst{p}\>x} \lor \neglit{\cst{q}\>(\cst{h}\>\cst{p}\>(\cst{p}\>\cst{b}))}
\}$. PLE does not apply because $\cst{p}$ and $\cst{q}$ occur with both polarities.
In addition, $\cst{p}$ also occurs unapplied and deep within a term.
Yet each clause contains either
a positive $\cst{p}$-literal or a negative $\cst{q}$-literal. Thus $\III
\models N''$, where $\III$ is as in Example~\ref{ex:ple}.
The additional literals are harmless.
\end{exa}

Examples \ref{ex:sle-i} and \ref{ex:sle-ii} suggest that pure literals are a
needlessly restrictive criterion in first- and higher-order logic. We propose a
generalization to ``quasipure literals.''
Although we present the criterion in a higher-order setting, it is equally
applicable for first-order logic. (In contrast, it is uninteresting for
propositional logic, because the only clauses with the same predicates, or
rather variables, with opposite polarities are tautologies, and these would be
deleted before pure literal elimination is attempted.)

\begin{defi}
\label{def:quasipure-set}
A \emph{polarity map} is a function that maps each predicate symbol in $\Sigma$
to a polarity ($+$ or $-$).
A set $P$ of predicate symbols is \emph{quasipure} in a clause set $N$ with a
polarity map $m$
if for every clause in $N$ that contains an
element of $P$, there exists a predicate symbol $\cst{p} \in P$ such that the
clause contains a $\cst{p}$-literal with polarity $m_\cst{p}$.
The set $P$ is \emph{quasipure} in $N$ if there exists a polarity map $m$
such that $P$ is quasipure in $N$ with $m$.
\end{defi}

In Example~\ref{ex:ple}, $\{\cst{p}, \cst{q}\}$ is quasipure in $N$ with $m_\cst{p} =
{+}$ and $m_\cst{q} = {-}$.
In Example~\ref{ex:sle-i}, $\{\cst{p}\}$ is quasipure in $N'$ with $m_\cst{p} = {+}$. In
Example~\ref{ex:sle-ii}, $\{\cst{p}, \cst{q}\}$ is quasipure in $N''$ with $m_\cst{p} =
{+}$ and $m_\cst{q} = {-}$. For this last example, it is crucial to consider
$\cst{p}$ and $\cst{q}$ together; neither of the singletons $\{\cst{p}\}$ and
$\{\cst{q}\}$ is quasipure in $N''$.

\begin{defi}
\label{def:quasipure-sym}
A predicate symbol $\cst{p}$ is \emph{quasipure} in a clause set $N$
with polarity $s \in \{+, -\}$
if there exists a set $P$ of predicate
symbols with $\cst{p} \in P$ and a polarity map $m$ such that $m_\cst{p} = s$ and
$P$ is quasipure in $N$ with $m$.
The symbol $\cst{p}$ is \emph{quasipure} in $N$ if there exists a polarity
$m_\cst{p} \in \{+, -\}$ such that $\cst{p}$ is quasipure in $N$ with $m_\cst{p}$.
A literal $L = \arbpredlit{\cst{p}\,\ldots}$ is \emph{quasipure} in $N$ if
$\cst{p}$ is quasipure in $N$ with $L$'s polarity.
\end{defi}

Notice that a predicate symbol that does not occur in a clause set is trivially
quasipure in that clause set.

Deleting a clause containing a quasipure literal might create new opportunities
for quasipure literal elimination, but it never ruins existing ones.
Therefore, the following nondeterministic procedure is confluent and hence
yields a unique result:

\begin{defi}
\label{def:sle}
Given a finite clause set, \emph{quasipure literal elimination} (QLE) repeatedly
removes clauses containing quasipure literals until no such literals remain.
\end{defi}

Although QLE is defined by iteration, it is always possible to remove all
clauses at the same time:

\begin{lem}
\label{lem:qle-single-swoop}
Let $N$ be a finite clause set and let $N'$ be the result of
QLE. Then there exists a predicate symbol set $P$ and a polarity map $m$
such that $P$ is quasipure in $N$ and for every clause in $N \setminus N'$
there exists a predicate symbol $\cst{q} \in P$ such that the
clause contains a $\cst{q}$-literal with polarity $m_\cst{q}$.
\end{lem}

\begin{proof}
The iterative process defining QLE gives rise to a finite sequence $(P_1,
m_1),\allowbreak \dotsc,\allowbreak (P_n, m_n)$ of predicate symbol sets and
polarity maps. Without loss of generality, we assume that each $P_j$ does not
contain predicate symbols that do not occur in the clause set at iteration
$j$. Then the sets $P_j$ are clearly mutually disjoint and we can take $P = P_1
\cup \cdots \cup P_n$ as the desired witness. As for the polarity map $m$, we
associate each $\cst{p} \in P_j$ with $m_j(\cst{p})$.
\end{proof}

The key property of QLE is that it preserves unsatisfiability:

\begin{lem}
Let $C \in N$ be a clause containing a quasipure literal. If $N \setminus \{C\}$
is satisfiable, then $N$ is satisfiable.
\end{lem}

\begin{proof}
Let $\III$ be a model of $N \setminus \{C\}$. Let $P$ be the set of predicate
symbols and $m$ the polarity map whose existence is guaranteed by
Lemma~\ref{lem:qle-single-swoop}.
Let $N_0 \subseteq N \setminus \{C\}$ be the result of
applying QLE on $N$. Clearly, $N_0$ contains no occurrences of the symbols in
$P$. Define $\III'$ based on $\III$ by redefining the semantics
of each monomorphic instance
$\cst{p}\tyargs{\tuple{\tau}} : \upsilon_1 \to \cdots \to \upsilon_n \to \omicron$
of symbol $\cst{p}$ such that $m_\cst{p} = {+}$ or $m_\cst{p} = {-}$: If
$m_\cst{p} = {+}$, interpret $\cst{p}\tyargs{\tuple{\tau}}$ as the predicate
that is uniformly true; otherwise, interpret $\cst{p}\tyargs{\tuple{\tau}}$ as
the predicate that is uniformly false. By the comprehension principle, both of
these predicates are guaranteed to exist in the interpretation of the type
$\upsilon_1 \to \cdots \to \upsilon_n \to \omicron$.
Now $\III'$ coincides with $\III$, since $N_0$ contains no $P$ symbols, and
thus $\III'$ is a model of $N_0$. In addition, $\III'$ is a model of $N$,
because each clause in $N \setminus N_0$ contains a quasipure literal, which is
satisfied by $\III'$.
\end{proof}

Definition~\ref{def:sle} suggests a naive, nondeterministic procedure for
discovering and eliminating quasipure predicate symbols: Choose a predicate
symbol $\cst{p}$ and a polarity $m_\cst{p}$, and take $P = \{\cst{p}\}$. If the
predicate symbol occurs with the wrong polarity in a clause
$\arbpredlit{\cst{p}\,\ldots} \lor \arbpredlit{\cst{q}_1\,\ldots} \lor \cdots \lor
\arbpredlit{\cst{q}_n\,\ldots} \lor C$, try to extend
the set $P$ with one of the $\arbpredlit{\cst{q}_i}$'s
and the polarity map $m$ accordingly, and continue
recursively with $\cst{q}_i$. In Section~\ref{sec:implementation}, we will see a
more efficient approach based on a SAT encoding.

\section{Implementation}
\label{sec:implementation}

We implemented the elimination techniques described in Sections
\ref{sec:hidden-literal-based-elimination}~to~\ref{sec:quasipure-literal-elimination}
in the Zipperposition prover. For HLBE, PE, and BCE, which had been studied by
Vukmirovi\'c et al., we could directly adapt their code
\cite[Section~6]{vukmirovic-et-al-2023-sat}. The data structure and algorithms
they described and implemented could be generalized to handle polymorphic
higher-order logic. For QLE, we developed our own code.

Zipperposition is a higher-order prover that implements the
$\lambda$-superposition calculus \cite{bentkamp-et-al-2021-hosup}, a
generalization of standard superposition to
classical rank-1 polymorphic higher-order logic---the logic described in
Section~\ref{sec:clausal-higher-order-logic}.
The prover is highly competitive: It won in the higher-order theorem
division of the CADE ATP Systems Competition (CASC)
\cite{sutcliffe-desharnais-2022-casc} in 2020, 2021, and 2022.

By default, Zipperposition \emph{immediately clausifies} the initial problem as
much as possible; then it optionally invokes preprocessing elimination
techniques. If inprocessing is enabled, the prover also invokes the elimination
techniques at regular intervals from within the saturation loop.
Immediate clausification performs well in practice, but an even more successful
strategy is to \emph{delay clausification}, interleaving clausification with
superposi\-tion-style calculus rules \cite[Section~4]{vukmirovic-et-al-2023-hot}.
Then it makes little sense to apply preprocessing elimination techniques;
inprocessing seems more appropriate.

\oursubsection{Hidden-Literal-Based Elimination}

HLBE relies on matching. In our setting, it needs to consider type variables and
higher-order terms. Our implementation uses an efficient approximation of
higher-order matching, which recognizes $\lambda y.\> y\>
\cst{a}$ as an instance of $\lambda y.\> y\> x$, but not $\cst{a}$ as an
instance of $y\>\cst{a}$ (with $y := \lambda x.\>x$). The same approximated
matching algorithm is used in Zipperposition to efficiently recognize subsumed
clauses. This weaker matching reduces the applicability of HLBE, but it does not
compromise its soundness.

\oursubsection{Predicate Elimination}

Compared with HLBE, more work was needed to make predicate elimination cope with
polymorphism and higher-order logic.
For polymorphism, the previous code simply did
not consider polymorphic predicates as predicates. Predicates needed to have a
type declaration of the form $\tau_1 \times \cdots \times \tau_n \to \omicron$.
Now, predicate symbols can have type declarations of the form
$\Pi\tuple{\alpha}.\>\tau_1 \to \cdots \to \tau_n \to \omicron$ as well as
$\Pi\tuple{\alpha}.\>\tau_1 \to \cdots \to \tau_n \to \alpha_i$;
via instantiation, $\alpha_i$ can become $\omicron$ or $\dots \to \omicron$.

For SPE, the type arguments of eliminated predicate symbols must be unified, as
per Definition~\ref{def:flat-resolvent}. For DPE, the type arguments must be
distinct type variables. In addition, for DPE, we must check that the predicate
is polymorphism-safe.

Adding support for higher-order logic required changing the definition of
singular predicates, which is used by both SPE and DPE, to check that the
predicate symbol to eliminate does not occur deep in the clauses that define the
symbol. In addition, for DPE, we need to synthesize a $\lambda$-abstraction
to replace any deep occurrences of the predicate outside the definition set.

\oursubsection{Blocked Clause Elimination}

Adding support for polymorphism to BCE was straightforward: We simply ensured
that type arguments are unified when computing flat resolvents and added a
polymorphism-safety check.

To support higher-order logic, we added a check that the $\cst{p}$-literal on
which the clause is resolved in the only $\cst{p}$-literal of that polarity
(corresponding to condition~\ref{itm:blocked-c'-no-same-pol} of
Definition~\ref{def:blocked}). We also disabled the code that computed $n$-ary
resolvents for $n > 2$. Finally, we added code to compute the list of all deep
predicate symbols $\cst{q}$ and made sure clauses are never blocked on a
$\cst{q}$-literal. These mechanisms come into play only if some higher-order
construct is detected in the input problem; otherwise, the first-order
formulation of BCE is used.

In the presence of equality in the logic, BCE relies on a congruence closure
algorithm to detect valid clauses \cite[Section~6.1]{ksstb-2017-blockedfol}. In
our implementation, we rely on a first-order congruence closure algorithm, which
can handle higher-order constructs but does not take advantage of them. For example,
resolving $\cst{p}\>\cst{a} \cor \lnot\> \cst{q}\>\cst{a}$ against
$\lnot\>\cst{p}\>\cst{b} \cor \cst{q}\>\cst{b}$ on the $\cst{p}$-literal yields
the clause $\cst{a} \noteq \cst{b} \cor \lnot\>\cst{q}\>\cst{a} \cor
\cst{q}\>\cst{b}$. Our congruence closure
algorithm can detect the validity of such a clause. On the other hand, because
the algorithm views $\lambda$-abstractions as black boxes, it fails
to recognize $\cst{a} \noteq \cst{b} \cor\allowbreak \lnot\>\cst{q}\>(\lambda
x.\>x\>\cst{a}) \cor\allowbreak \cst{q}\>(\lambda x.\> x\> \cst{b})$ as valid.

\oursubsection{Quasipure Literal Elimination}

A simple and efficient implementation of quasipure literal elimination uses a SAT
solver. Let $N$ be a finite clause set. Without loss of generality, we
can assume that $\Sigma$ is finite. The signature of the generated SAT problem consists of
the variables $\cst{p}^+, \cst{p}^-$ for each predicate symbol $\cst{p}$ in
$\Sigma$, where $\cst{p}^s$ means ``$\cst{p}$ (possibly together with
other predicate symbols) is quasipure with polarity~$s$.'' The SAT problem
consists of the following clauses:
\begin{enumerate}
\itemsep1\jot
\item For each clause in $N$ containing $n$ predicate literals, headed by
  $\cst{q}_1,\dots,\cst{q}_n$ and with respective polarities
  $s_1,\ldots,s_n$, generate $n$~clauses of the form
  \[\cst{q}_1^{s_1}\lor\cdots\lor\cst{q}_{j-1}^{s_{j-1}}
    \lor
    \lnot\,{\cst{q}_i^{-s_j}}
    \lor
    \cst{q}_{j+1}^{s_{j+1}}\lor\cdots\lor\cst{q}_n^{s_n}
   \]
   where $-s$ flips the polarity $s$. Such clauses ensure that whenever a
   literal $\arbpredlit{\cst{q}_i\,\ldots}$ has the wrong polarity according to
   the current variable assignment, one of the other predicate literals must be
   quasipure.

\item For each clause in $N$ containing $n$ predicate literals, headed by
  $\cst{q}_1,\dots,\cst{q}_n$ and with respective polarities
  $s_1,\ldots,s_n$ and containing a deep occurrence of
  $\cst{p}$ (in an argument to a $\cst{q}_j$ or in a functional literal),
  generate the two clauses
\begin{align*}
  & \lnot\,{\cst{p}^+} \lor \cst{q}_1^{s_1}\lor\cdots\lor\cst{q}_n^{s_n}
 && \lnot\,{\cst{p}^-} \lor \cst{q}_1^{s_1}\lor\cdots\lor\cst{q}_n^{s_n}
\end{align*}
   These clauses ensures that whenever $\cst{p}$ occurs deep and is
   nonetheless considered quasipure, one of the predicate literals must be
   quasipure.

\item For each predicate symbol $\cst{p}$ in $\Sigma$, generate the
  clause $\lnot\, \cst{p}^+ \lor \lnot\, \cst{p}^-$. It ensures that a single
  polarity is assigned to a quasipure predicate symbol.

\item Generate the clause $\cst{p}_1^+ \lor \cst{p}_1^- \lor \cdots \lor
  \cst{p}_n^+ \lor \cst{p}_n^-$, where $\{\cst{p}_1,\dots,\cst{p}_n\}$ are the
  predicate symbols in $\Sigma$. It tells the SAT solver to look for a
  nontrivial solution, in which at least one predicate symbol is quasipure.
\end{enumerate}

From a satisfying assignment, we can easily read off a predicate symbol set and a
polarity map. The process can be iterated until we reach a maximal solution,
at which point the SAT solver returns a verdict of ``unsatisfiable.''

It is unclear whether this problem is \textbf{NP}-complete.
Given the nondeterministic nature of the naive procedure, we suspect that it is,
but we have not found a reduction from 3-SAT. Thus, it is unclear whether our use
of a SAT solver is fully satisfactory from a theoretical point of view, even
though it works well in practice.

\section{Evaluation}
\label{sec:evaluation}

We evaluate the techniques presented above by running Zipperposition
\cite{bentkamp-et-al-2021-hosup} on various benchmarks in various
configurations. We consider seven benchmark sets:
\begin{itemize}
\item \emph{S0}:\ a randomly selected subset of 1000 higher-order monomorphic (TH0)
  problems from the Sledgehammer-generated Seventeen benchmark suite
  \cite{desharnais-et-al-2022} in the base configuration, in the language
  fragment called TH0$^-$ in the Seventeen paper;

\smallskip

\item \emph{S1}:\ a randomly selected subset of 1000 higher-order polymorphic (TH1)
  problems from the Seventeen benchmark suite in the base configuration, in the
  language fragment called TH1$^-$ in the Seventeen paper;

\smallskip

\item \emph{TH0}:\ a randomly selected subset of 1000 higher-order monomorphic (TH0)
  problems from the TPTP \cite{sutcliffe-2017-tptp} version 8.0.0;

\smallskip

\item \emph{TH1}:\ a randomly selected subset of 1000 higher-order polymorphic (TH1)
  problems from the TPTP version 8.0.0;

\smallskip

\item \emph{CF}:\ a predefined set of 1000 first-order untyped (CNF and FOF)
  problems from the TPTP;

\smallskip

\item \emph{TF0}:\ all 389 first-order monomorphic (TF0) problems without
  arithmetic from the TPTP;

\smallskip

\item \emph{TF1}:\ all 678 first-order polymorphic (TF1) problems without
  arithmetic from the TPTP.
\end{itemize}
The first four benchmark sets are used to determine how much our techniques can
help on higher-order problems. Among these, S1 and TH1 exercise polymorphism. As
for the remaining three, they are included for comparison; they show how well
the techniques work on first-order problems. The benchmark sets contain some
problems known to be unprovable, which we use to check soundness of the
techniques.

We consider 12 Zipperposition configurations, all derived from the portfolio
of time slices that was used at the 2022 edition of CASC. This portfolio does
not use any of our techniques. Even on first-order problem, it applies the
higher-order $\lambda$-superposition calculus. The 12~configurations are as
follows:

\begin{itemize}
\item \emph{None}:\ the baseline, corresponding to the portfolio used at CASC;

\smallskip

\item $X$\emph{-pre}:\ the baseline modified to use technique $X \in
  \{\text{PE}, \text{BCE}, \text{PLE}, \text{QLE}\}$ as a preprocessor
  in all the time slices;

\smallskip

\item $X$\emph{-in}:\ the baseline modified to use technique $X \in
  \{\text{HLBE}, \text{PE}, \text{BCE}, \text{PLE},\allowbreak \text{QLE}\}$ as an
  inprocessor in all the time slices;

\smallskip

\item \emph{All-pre}:\ the baseline modified to use all of
  \text{PE}, \text{BCE}, and \text{QLE} as preprocessors in all the time slices;

\smallskip

\item \emph{All-in}:\ the baseline modified to use all of
  HLBE, \text{PE}, \text{BCE}, and \text{QLE} as inprocessors in all the time
  slices.
\end{itemize}
In addition, we define the virtual configuration \emph{Union}, consisting of
the virtual portfolio of all other 12~configurations. All problems that are
solved in at least one of the 12~configurations are considered solved by Union,
and only those.

The experiments were carried out on StarExec Miami \cite{sst-2014-starexec}
servers equipped with Intel Xeon E5-2620 v4 CPUs clocked at 2.10 GHz. We used
CPU and wallclock time limits of 120~s.
The raw results are available online.\footnote{\url{https://zenodo.org/record/6997515}}

Figure~\ref{fig:solved-problems} reports how many problems were solved for each
combination of benchmark set and configuration. The last row of the table
presents the total of the seven rows above it. Bold singles out the best
configuration (other than Union) for each benchmark set.

The results are sobering. We see substantial gains on the untyped first-order
benchmarks, but the gains are much more modest, if actually present, on the
typed first-order and the higher-order benchmarks. A possible explanation is
that there is less clausal structure in a higher-order problem. Most of
Zipperposition's time slices clausify the problem lazily, meaning that little
information is visible to our techniques, especially when used as preprocessors.
Another possible explanation might be that the higher-order Seventeen and TPTP
benchmarks look quite different from the untyped first-order TPTP benchmarks,
and our techniques are less applicable. For example, the higher-order TPTP
problems tend to be much smaller than their first-order counterparts. Moreover,
we selected only subsets of the TPTP benchmarks---an evaluation on entire
benchmark suites might yield different results.

The picture is more positive if we look at the Union column of the table.
Clearly, in a portfolio setting, with enough time, the new techniques can make a
useful contribution.

We also notice that inprocessing often performs worse than preprocessing.
This corroborates the findings of Vukmirovi\'c et
al.~\cite{vukmirovic-et-al-2023-sat}. An explanation might be the heavy cost of
running the techniques multiple time, during proof search. In addition, PE, BCE,
PLE, and QLE rely on a global analysis of the clause set and tend to
become less applicable as the clause set grows.

Finally, we see that PLE and QLE help, especially on first-order problems.
Unexpectedly, PLE generally outperforms the more general QLE.

\newcommand\win[1]{\phantom{#1}\llap{\bf{#1}\kern-.2ex}}

\begin{figure}[t]
\centering
\begin{tabular}{@{}l@{\kern.8em}c@{\kern.8em}c@{\kern.8em}c@{\kern.4em}c@{\kern.8em}c@{\kern.4em}c@{\kern.8em}c@{\kern.4em}c@{\kern.8em}c@{\kern.4em}c@{\kern.8em}c@{\kern.4em}c@{\kern.4em}c@{}}
& None & \!HLBE\! & \multicolumn{2}{c}{PE} & \multicolumn{2}{c}{BCE} & \multicolumn{2}{c}{PLE} & \multicolumn{2}{c}{QLE} & \multicolumn{2}{c}{All} & Union \\
& & in & pre & in & pre & in & pre & in & pre & in & pre & in \\
\midrule
\strut S0 & \win{\phantom{0}575} & \phantom{0}574 & \phantom{0}574 & \phantom{0}572 & \win{\phantom{0}575} & \phantom{0}574 & \phantom{0}574 & \win{\phantom{0}575} & \win{\phantom{0}575} & \phantom{0}574 & \phantom{0}572 & \win{\phantom{0}575} & \phantom{0}589 \\
\strut S1 & \phantom{0}349 & \phantom{0}343 & \phantom{0}351 & \phantom{0}350 & \phantom{0}347 & \phantom{0}348 & \phantom{0}352 & \phantom{0}352 & \phantom{0}352 & \phantom{0}353 & \win{\phantom{0}354} & \phantom{0}344 & \phantom{0}366 \\
\strut TH0 & \phantom{0}711 & \phantom{0}713 & \phantom{0}713 & \phantom{0}709 & \phantom{0}713 & \win{\phantom{0}714} & \phantom{0}710 & \phantom{0}713 & \phantom{0}711 & \phantom{0}705 & \phantom{0}711 & \phantom{0}706 & \phantom{0}722 \\
\strut TH1 & \phantom{0}338 & \phantom{0}334 & \phantom{0}339 & \phantom{0}340 & \win{\phantom{0}341} & \phantom{0}339 & \phantom{0}337 & \win{\phantom{0}341} & \phantom{0}337 & \phantom{0}336 & \phantom{0}338 & \phantom{0}336 & \phantom{0}354 \\
\strut CF & \phantom{0}504 & \phantom{0}513 & \phantom{0}512 & \phantom{0}510 & \phantom{0}509 & \phantom{0}507 & \phantom{0}513 & \phantom{0}505 & \phantom{0}508 & \phantom{0}505 & \phantom{0}517 & \win{\phantom{0}519} & \phantom{0}545 \\
\strut TF0 & \phantom{0}138 & \win{\phantom{0}142} & \phantom{0}138 & \phantom{0}137 & \win{\phantom{0}142} & \phantom{0}139 & \phantom{0}141 & \phantom{0}141 & \phantom{0}139 & \phantom{0}141 & \phantom{0}139 & \phantom{0}137 & \phantom{0}148 \\
\strut TF1 & \phantom{0}227 & \phantom{0}227 & \phantom{0}229 & \win{\phantom{0}231} & \phantom{0}228 & \phantom{0}230 & \phantom{0}229 & \phantom{0}230 & \phantom{0}229 & \phantom{0}230 & \phantom{0}228 & \phantom{0}224 & \phantom{0}234 \\
[1\jot]
\strut Total & 2842 & 2846 & 2856 & 2849 & 2855 & 2851 & 2856 & 2857 & 2851 & 2844 & \win{2859} & 2841 & 2958 \\
\end{tabular}
\caption{Number of solved problems per benchmark set and configuration}
\label{fig:solved-problems}
\end{figure}

\section{Conclusion}
\label{sec:conclusion}

We presented four SAT-inspired techniques for transforming higher-order problems
with the aim of making them more amenable to automatic proof search. Three of
the techniques (HLBE, PE, and BCE) had been previously generalized to
first-order logic; we now generalized them further to higher-order logic. The
fourth technique (QLE) is new.

On the theoretical side, we showed that the techniques preserve satisfiability
and unsatisfiability of problems with respect to Henkin semantics. On the
practical side, we implemented the techniques in the higher-order prover
Zipperposition. Regrettably, the techniques did not perform as well on
higher-order problems as they do on first-order problems. This could be due to
the nature of the benchmark sets.

\let\Acksize=\relax

\def\ackname{\Acksize Acknowledgment}
\paragraph{\textbf{\upshape\ackname.}}
{\Acksize
Alexander Bentkamp gave us some advice about general interpretations.
The anonymous reviewers made many useful suggestions.

This research has received funding from
the European Research Council (ERC) under the European Union's Horizon 2020
research and innovation program (grant agreement No.\ 713999, Matryoshka).
The research has also
received funding from the Netherlands Organization for Scientific Research (NWO)
under the Vidi program (project No.\ 016.Vidi.189.037, Lean Forward).

}

\bibliographystyle{alphaurl}
\bibliography{main} 

\end{document}